\documentclass[11pt,draftcls,onecolumn]{IEEEtran}

\usepackage{amsmath,amsfonts} \usepackage{theorem}
\usepackage{graphicx,epsfig,color}
\usepackage{algorithmic}

\newtheorem{theorem}{Theorem}[section]
\newtheorem{proposition}[theorem]{Proposition}

\newtheorem{definition}[theorem]{Definition}
\newtheorem{lemma}[theorem]{Lemma} 

\newtheorem{assumption}{Assumption}[section]
{\theorembodyfont{\rmfamily} 
\newtheorem{remark}[theorem]{Remark}

\newtheorem{example}[theorem]{Example}}
\newtheorem{problem}[theorem]{Problem}

\newcommand{\real}{\mathbb{R}}
\renewcommand{\natural}{\mathbb{N}}

\newcommand{\until}[1]{\{1,\dots, #1\}}
\newcommand{\subscr}[2]{#1_{\textup{#2}}}

\newcommand{\setdef}[2]{\{#1 \; | \; #2\}}

\newcommand{\map}[3]{#1: #2 \rightarrow #3}

\newcommand\oprocendsymbol{\hbox{$\square$}}
\newcommand\oprocend{\relax\ifmmode\else\unskip\hfill\fi\oprocendsymbol}

\newcommand{\realpositive}{\ensuremath{\mathbb{R}}_{>0}}

\def\beq{\begin{equation}}
\def\eeq{\end{equation}}
\def\Z{\mathbb{Z}}
\def\N{\mathbb{N}}
\def\E{\mathbb{E}}
\def\R{\mathbb{R}}

\def\P{\mathbb{P}}

\def\0{\mathbf{0}}
\def\1{\mathbf{1}}

\newcommand{\pij}{\subscr{p}{$i|j$}}
\newcommand{\pjj}{\subscr{p}{$j|j$}}

\newcommand{\poz}{\subscr{p}{$1|0$}}
\newcommand{\poo}{\subscr{p}{$1|1$}}
\newcommand{\pzz}{\subscr{p}{$0|0$}}
\newcommand{\pzo}{\subscr{p}{$0|1$}}

\newcommand{\poj}{\subscr{p}{$1|j$}}
\newcommand{\pzj}{\subscr{p}{$0|j$}}

\newcommand{\pio}{\subscr{p}{$i|1$}}
\newcommand{\poi}{\subscr{p}{$i|1$}}
\newcommand{\piz}{\subscr{p}{$i|0$}}

\newcommand{\pco}{\subscr{p}{c$|1$}}
\newcommand{\pwo}{\subscr{p}{w$|1$}}
\newcommand{\pndj}{\subscr{p}{nd$|j$}}
\newcommand{\pcj}{\subscr{p}{c$|j$}}
\newcommand{\pwj}{\subscr{p}{w$|j$}}
\newcommand{\pwz}{\subscr{p}{w$|0$}}

\newcommand{\pwjf}{\subscr{p}{w$|j$}^{(\textup{f})}}
\newcommand{\pwof}{\subscr{p}{w$|1$}^{(\textup{f})}}

\newcommand{\pcjf}{\subscr{p}{c$|j$}^{(\textup{f})}}

\newcommand{\pwjm}{\subscr{p}{w$|j$}^{(\textup{m})}}
\newcommand{\pwom}{\subscr{p}{w$|1$}^{(\textup{m})}}

\newcommand{\pcjm}{\subscr{p}{c$|$j}^{(\textup{m})}}

\def\pzo{\subscr{p}{$0|1$}}
\def\pndo{\subscr{p}{nd$|1$}}
\def\pndz{\subscr{p}{nd$|0$}}

\newcommand{\ceil}[1]{\left\lceil #1\right\rceil}
\newcommand{\floor}[1]{\left\lfloor #1\right\rfloor}

\newcommand{\Bigceil}[1]{\Big\lceil#1\Big\rceil}

\title{Accuracy and Decision Time \\ for Sequential Decision
  Aggregation\thanks{This work has been supported in part by AFOSR MURI
    FA9550-07-1-0528.}}

\author{Sandra H. Dandach \qquad\and\quad Ruggero Carli \qquad\and\quad Francesco Bullo
  \thanks{S. H. Dandach and R. Carli and F. Bullo are with the Center for Control, Dynamical
    Systems and Computation, University of California at Santa Barbara,
    Santa Barbara, CA 93106, USA, {\tt\small
      \{sandra|carlirug|bullo\}@engr.ucsb.edu}.  }}

\begin{document}
\maketitle  

\begin{abstract}
  This paper studies prototypical strategies to sequentially aggregate
  independent decisions.  We consider a collection of agents, each
  performing binary hypothesis testing and each obtaining a decision over
  time.  We assume the agents are identical and receive independent
  information.  Individual decisions are sequentially aggregated via a
  threshold-based rule. In other words, a collective decision is taken as
  soon as a specified number of agents report a concordant decision
  (simultaneous discordant decisions and no-decision outcomes are also
  handled).

  We obtain the following results.  First, we characterize the
  probabilities of correct and wrong decisions as a function of time, group
  size and decision threshold. The computational requirements of our
  approach are linear in the group size.  Second, we consider the so-called
  fastest and majority rules, corresponding to specific decision
  thresholds. For these rules, we provide a comprehensive scalability
  analysis of both accuracy and decision time.  In the limit of large group
  sizes, we show that the decision time for the fastest rule converges to
  the earliest possible individual time, and that the decision accuracy for
  the majority rule shows an exponential improvement over the individual
  accuracy.  Additionally, via a theoretical and numerical analysis, we
  characterize various speed/accuracy tradeoffs. Finally, we relate our
  results to some recent observations reported in the cognitive information
  processing literature.
\end{abstract}

\section{Introduction}

\subsection{Problem setup}
Interest in group decision making spans a wide variety of domains. Be it in
electoral votes in politics, detection in robotic and sensor networks, or
cognitive data processing in the human brain, establishing the best
strategy or understanding the motivation behind an observed strategy, has
been of interest for many researchers.  This work aims to understand how
grouping individual sequential decision makers affects the speed and
accuracy with which these individuals reach a collective decision.  This
class of problems has a rich history and some of its variations are studied
in the context of distributed detection in sensor networks and Bayesian
learning in social networks.

In our problem, a group of individuals independently decide between two
alternative hypothesis, and each individual sends its local decision to a
fusion center. The fusion center decides for the whole group as soon as one
hypothesis gets a number of votes that crosses a pre-determined
threshold. We are interested in relating the accuracy and decision time of
the whole population, to the accuracy and decision time of a single
individual. We assume that all individuals are independent and
identical. That is, we assume that they gather information corrupted by
i.i.d.\ noise and that the same statistical test is used by each individual
in the population.  The setup of similar problems studied in the literature
usually assumes that all individual decisions need to be available to the
fusion center, before the latter can reach a final decision. The work
presented here relaxes this assumption and the fusion center might provide
the global decision much earlier than the all individuals in the group.
Researchers in behavioral studies refer to decision making schemes where
everyone is given an equal amount of time to respond as the ``free response
paradigm.''  Since the speed of the group's decision is one of our main
concerns, we adjust the analysis in a way that makes it possible to compute
the joint probabilities of each decision at each time instant. Such a
paradigm is referred to as the ``interrogation paradigm.''

\subsection{Literature review}

The framework we analyze in this paper is related to the one considered in
many papers in the literature, see for instance \cite{JNT:93, PkV:96,
  WWI-JNT:94, JNT:84, LY:94, VVV-TB-HVP:94, VVV:01, DA-MAD-IL-AO:08-report,
  ATS-AS-AJ:09} and references therein. The focus of these works is mainly
two-fold. First, researchers in the fields aim to determine which type of
information the decision makers should send to the fusion center.  Second,
many of the studies concentrate on computing optimal decision rules both
for the individual decision makers and the fusion center where optimality
refers to maximizing accuracy. One key implicit assumption made in numerous
works, is that the aggregation rule is applied by the fusion center only
after all the decision makers have provided their local decisions.

Tsitsiklis in~\cite{JNT:93} studied the Bayesian decision problem with a
fusion center and showed that for large groups identical local decision
rules are asymptotically optimal. Varshney in~\cite{PkV:96} proved that
when the fusion rules at the individuals level are non-identical, threshold
rules are the optimal rules at the individual level.  Additionally,
Varshney proved that setting optimal thresholds for a class of fusion
rules, where a decision is made as soon as a certain number $q$ out of the
$N$ group members decide, requires solving a number of equations that grows
exponentially with the group size. The fusion rules that we study in this
work fall under the $q$ out of $N$ class of decision rules. Finally,
Varshney proved that this class of decision rules is optimal for identical
local decisions.

\subsection{Contributions}
The contributions of this paper are three-folds.  First, we introduce a
recursive approach to characterize the probabilities of correct and wrong
decisions for a group of sequential decision makers (SDMs). These
probabilities are computed as a function of time, group size and decision
threshold.  The key idea is to relate the decision probability for a group
of size $N$ at each time $t$, to the decision probability of an individual
SDM up to that time $t$, in a recursive manner. Our proposed method has
many advantages. First, our method has a numerical complexity that grows
only linearly with the number of decision makers. Second, our method is
independent of the specific decision making test adopted by the SDMs and
requires knowledge of only the decision probabilities of the SDMs as a
function of time. Third, our method allows for asynchronous decision times
among SDMs.  To the best of our knowledge, the performance of sequential
aggregation schemes for asynchronous decisions has not been previously
studied.

Second, we consider the so-called \emph{fastest} and \emph{majority} rules
corresponding, respectively, to the decision thresholds $q=1$ and
$q=\ceil{N/2}$. For these rules we provide a comprehensive scalability
analysis of both accuracy and decision time.  Specifically, in the limit of
large group sizes, we provide exact expressions for the expected decision
time and the probability of wrong decision for both rules, as a function of
the decision probabilities of each SDM.  For the \emph{fastest} rule we
show that the group decision time converges to the earliest possible
decision time of an individual SDM, i.e., the earliest time for which the
individual SDM has a non-zero decision probability.  Additionally, the
\emph{fastest} rule asymptotically obtains the correct answer almost
surely, provided the individual SDM is more likely to make the correct
decision, rather than the wrong decision, at the earliest possible decision
time.  For the \emph{majority} rule we show that the probability of wrong
decision converges exponentially to zero if the individual SDM has a
sufficiently small probability of wrong decision.  Additionally, the
decision time for the \emph{majority} rule is related to the earliest time
at which the individual SDM is more likely to give a decision than to not
give a decision.  This scalability analysis relies upon novel asymptotic
and monotonicity results of certain binomial expansions.

As third main contribution, using our recursive method, we present a
comprehensive numerical analysis of sequential decision aggregation based
on the \emph{$q$ out of $N$} rules. As model for the individual SDMs, we
adopt the sequential probability ratio test (SPRT), which we characterize
as an absorbing Markov chain.  First, for the \emph{fastest} and
\emph{majority} rules, we report how accuracy and decision time vary as a
function of the group size and of the SPRT decision probabilities.  Second,
in the most general setup, we report how accuracy and decision time vary
monotonically as a function of group size and decision threshold.
Additionally, we compare the performance of fastest versus majority rules,
at fixed group accuracy.  We show that the best choice between the fastest
rule and the majority rule is a function of group size and group accuracy.
Our numerical results illustrate why the design of optimal aggregation
rules is a complex task~\cite{JNT-MA:85}.  Finally, we discuss
relationships between our analysis of sequential decision aggregation and
mental behavior documented in the cognitive psychology and neuroscience
literature~\cite{FC-JC-RV-GD:09, SW-UN:09, SW-UN:10,PL-TP-TS-MW-BS:05}.

Finally, we draw some qualitative lessons about sequential decision
aggregation from our mathematical analysis.  Surprisingly, our results show
that the accuracy of a group is not necessarily improved over the accuracy
of an individual.  In aggregation based on the \emph{majority} rule, it is
true that group accuracy is (exponentially) better than individual
accuracy; decision time, however, converges to a constant value for large
group sizes.  Instead, if a quick decision time is desired, then the
\emph{fastest} rule leads, for large group sizes, to decisions being made
at the earliest possible time.  However, the accuracy of fastest
aggregation is not determined by the individual accuracy (i.e., the time
integral of the probability of correct decision over time), but is rather
determined by the individual accuracy at a specific time instant, i.e., the
probability of correct decision at the earliest decision time.  Accuracy at
this special time might be arbitrarily bad especially for "asymmetric"
decision makers (e.g., SPRT with asymmetric thresholds).  Arguably, these
detailed results for \emph{fastest} and \emph{majority} rules, $q=1$ and
$q=\floor{N/2}$ respectively, are indicative of the accuracy and decision
time performance of aggregation rules for small and large thresholds,
respectively.

\subsection{Decision making in cognitive psychology}\label{ssection-cog}

An additional motivation to study sequential decision aggregation is our
interest in sensory information processing systems in the brain.
There is a growing belief among
neuroscientists~\cite{SW-UN:09,SW-UN:10,PL-TP-TS-MW-BS:05} that the brain
normally engages in an ongoing synthesis of streams of information
(stimuli) from multiple sensory modalities. Example modalities include
vision, auditory, gustatory, olfactory and somatosensory.
While many areas of the brain (e.g., the primary projection pathways)
process information from a single sensory modality, many nuclei (e.g., in
the Superior Colliculus) are known to receive and integrate stimuli from
multiple sensory modalities.  Even in these multi-modal sites, a specific
stimulus might be dominant.
Multi-modal integration is indeed relevant when the response elicited by
stimuli from different sensory modalities is statistically different from
the response elicited by the most effective of those stimuli presented
individually.  (Here, the response is quantified in the number of impulses from
neurons.)
Moreover, regarding data processing in these multi-modal sites, the
procedure with which stimuli are processed changes depending upon the
intensity of each modality-specific stimulus.
 
In~\cite{SW-UN:09}, Werner et al.\ study a human decision making problem
with multiple sensory modalities. They present examples where accuracy and
decision time depend upon the strength of the audio and visual components
in audio-visual stimuli.
They find that, for intact stimuli (i.e., noiseless signals), the decision
time improves in multi-modal integration (that is, when both stimuli are
simultaneously presented) as compared with uni-sensory integration.
Instead, when both stimuli are degraded with noise, multi-modal
integration leads to an improvement in both accuracy and decision time.
Interestingly, they also identify circumstances for which multi-modal
integration leads to performance degradation: performance with an intact
stimulus together with a degraded stimulus is sometimes worse than
performance with only the intact stimulus.

Another point of debate among cognitive neuroscientists is how to
characterize uni-sensory versus multi-modal integration sites.
Neuro-physiological studies have traditionally classified as multi-modal
sites where stimuli are enhanced, that is, the response to combined stimuli
is larger than the sum of the responses to individual stimuli.  Recent
observations of suppressive responses in multi-modal sites has put this
theory to doubt; see~\cite{SW-UN:10,PL-TP-TS-MW-BS:05} and references
therein. More specifically, studies have shown that by manipulating the
presence and informativeness of stimuli, one can affect the performance
(accuracy and decision time) of the subjects in interesting, yet not well
understood ways.  We envision that a more thorough theoretical
understanding of sequential decision aggregation will help bridge the gap
between these seemingly contradicting characterization of multi-modal
integration sites.

As a final remark about uni-sensory integration sites, it is well
known~\cite{RB-EB-etal:06} that the cortex in the brain integrates
information in \emph{neural groups} by implementing a \emph{drift-diffusion
  model}. This model is the continuous-time version of the so-called
sequential probability ratio test (SPRT) for binary hypothesis testing.
We will adopt the SPRT model for our numerical results.

\subsection{Organization}
We start in Section~\ref{Sec-Prob-setup} by introducing the problem
setup. In Section~\ref{Sec-analysis-nocomm} we present the numerical method
that allows us to analyze the decentralized Sequential Decision Aggregation
(SDA) problem; We analyze the two proposed rules in
Section~\ref{sec-asymp-analysis}. We also present the numerical results in
Section~\ref{SSec-sim-nocomm}. Our conclusions are stated in
Section~\ref{Sec-conc}. The appendices contain some results on binomial
expansions and on the SPRT.

\section{ Models of sequential aggregation and problem statement}\label{Sec-Prob-setup}

In this section we introduce the model of sequential aggregation and the
analysis problem we want to address. Specifically in
Subsection~\ref{subsec:DM} we review the classical sequential binary
hypothesis testing problem and the notion of \emph{sequential decision
  maker}, in Subsection~\ref{subsec:DHT} we define the \emph{$q$ out of $N$
  sequential decisions aggregation} setting and, finally, in
Subsection~\ref{subsec:PF}, we state the problem we aim to solve.

\subsection{Sequential decision maker}\label{subsec:DM}

The classical binary sequential decision problem is posed as follows.

Let $H$ denote a hypothesis which takes on values $H_0$ and $H_1$. Assume
we are given an individual (called \emph{sequential decision maker (SDM)}
hereafter) who repeatedly observes at time $t=1,2,\ldots,$ a random
variable $X$ taking values in some set $\mathcal{X}$ with the purpose of
deciding between $H_0$ and $H_1$.  Specifically the SDM takes the
observations $x(1), x(2), x(3), \ldots$, until it provides its final
decision at time $\tau$, which is assumed to be a stopping time for the
sigma field sequence generated by the observations, and makes a final
decision $\delta$ based on the observations up to time $\tau$. The stopping
rule together with the final decision rule represent the decision policy of
the SDM.  The standing assumption is that the conditional joint
distributions of the individual observations under each hypothesis are
known to the SDM.

In our treatment, we do not specify the type of decision policy adopted by
the SDM. A natural way to keep our presentation as general as possible, is
to refer to a probabilistic framework that conveniently describes the
sequential decision process generated by any decision policy. Specifically,
given the decision policy $\gamma$, let $\chi_0^{(\gamma)}$ and
$\chi_1^{(\gamma)}$ be two random variables defined on the sample space $\N
\times \{0,1\} \cup \{ ? \}$ such that, for $i,j \in \{0,1\}$,
\begin{itemize}
\item $\{\chi_j^{(\gamma)}=(t, i)\}$ represents the event that the individual decides in favor of $H_i$ at time $t$ given that the true hypothesis is $H_j$; and
\item $\{\chi_j^{(\gamma)}=?\}$ represents the event that the individual never reaches a decision given that $H_j$ is the correct hypothesis.
\end{itemize}

Accordingly, define $\pij^{(\gamma)}(t)$ and $\pndj^{(\gamma)}$ to be the probabilities that, respectively, the events $\{\chi^{(\gamma)}_j=(t, i)\}$ and $\{\chi^{(\gamma)}_0=?\}$ occur, i.e,
$$\pij^{(\gamma)}(t) = \P[\chi_j^{(\gamma)} = (t,i)]
\qquad \text{and}  \qquad \pndj^{(\gamma)} = \P[\chi_j^{(\gamma)} = ?].$$ 

Then the sequential decision process induced by the decision policy $\gamma$ is completely characterized by the following two sets of probabilities
\begin{equation}\label{eq:SDMpro}
\left\{\pndz^{(\gamma)}\right\} \cup \left\{ \pzz^{(\gamma)} (t), \poz^{(\gamma)}(t)\right\}_{t\in \N} \qquad \text{and} \qquad  \left\{\pndo^{(\gamma)}\right\} \cup \left\{\pzo^{(\gamma)}(t), \poo^{(\gamma)}(t)\right\}_{t\in \N},
\end{equation}
where, clearly $\pndz^{(\gamma)} + \sum_{t=1}^{\infty} \left(\pzz^{(\gamma)}(t)+\poz^{(\gamma)}(t)\right)=1$ and $\pndo^{(\gamma)} + \sum_{t=1}^{\infty} \left(\pzo^{(\gamma)}(t)+\poo^{(\gamma)}(t)\right)=1$. In what follows, while referring to a SDM running a sequential distributed hypothesis test with a pre-assigned decision policy, we will assume that the above two probabilities sets are known. From now on, for simplicity, we will drop the superscript $(\gamma)$.

Together with the probability of no-decision, for $j\in \{0,1\}$ we
introduce also the probability of correct decision
$\pcj:=\P[\text{say }H_j\,|\, H_j]$ and the probability of
wrong decision $\pwj:=\P[\text{say }H_i, \,i\neq j\,|\,H_j]$,
that is,
$$
\pcj=\sum_{t=1}^{\infty} \pjj(t) \qquad \text{and} \qquad \pwj=\sum_{t=1}^{\infty} \pij(t), \,\, i\neq j.
$$ 
It is worth remarking that in most of the binary sequential decision
making literature, $\pwo$ and $\pwz$ are referred as, respectively, the \emph{mis-detection} and \emph{false-alarm} probabilities of error.

Below, we provide a formal definition of two properties that the SDM might or might not satisfy.
\begin{definition}
  For a SDM with decision probabilities as in~\eqref{eq:SDMpro}, the
  following properties may be defined:
  \begin{enumerate}
  \item the SDM has \emph{almost-sure decisions} if, for $j\in\{0,1\}$,
    $$
    \sum_{t=1}^\infty \left(\pzj(t)+ \poj(t)\right)=1, \quad
    \text{and}
    $$
  \item the SDM has \emph{finite expected decision time} if, for $j\in\{0,1\}$,
    \begin{equation*}
      \sum_{t=1}^{\infty} t \left(\pzj(t)+ \poj(t)\right) < \infty.
    \end{equation*}
    \end{enumerate}
\end{definition}
One can show that the finite expected decision time implies almost-sure
decisions.

We conclude this section by briefly discussing examples of sequential
decision makers. The classic model is the SPRT model, which we discuss in
some detail in the example below and in Section~\ref{SSec-sim-nocomm}.  Our
analysis, however, allows for arbitrary sequential binary hypothesis tests,
such as the SPRT with time-varying thresholds~\cite{YL-SB:92}, constant
false alarm rate tests~\cite{SK-LLS:99}, and generalized likelihood ratio
tests.  Response profiles arise also in neurophysiology, e.g.,
\cite{TO-YI-MT-TF:07} presents neuron models with a response that varies
from unimodal to bimodal depending on the strength of the received
stimulus.

\begin{example}[Sequential probability ratio test (SPRT)] 
  In the case the observations taken are independent, conditioned on each
  hypothesis, a well-known solution to the above binary decision problem is
  the so-called \emph{sequential probability ratio test (SPRT)} that we
  review in Section~\ref{SSec-sim-nocomm}. A SDM implementing the SPRT test
  has both the \emph{almost-sure decisions} and \emph{finite expected
    decision time} properties.  Moreover the SPRT test satisfies the
  following optimality property: among all the sequential tests having
  pre-assigned values of \emph{mis-detection} and \emph{false-alarm}
  probabilities of error, the SPRT is the test that requires the smallest
  expected number of iterations for providing a solution.

  In Appendices~\ref{subsec:AccTimeDiscrete} and~\ref{subsec:AccTimeCont}
  we review the methods proposed for computing the probabilities
  $\left\{\pij(t)\right\}_{t\in \N}$ when the SPRT test is applied, both in
  the case $X$ is a discrete random variable and in the case $X$ is a
  continuous random variable. For illustration purposes, we provide in
  Figure~\ref{fig-pij} the probabilities $\pij(t)$ when $j=1$ for the case
  when $X$ is a continuous random variable with a continuous
  distribution (Gaussian). We also note that $p_{i|j}(t)$ might have
  various interesting distributions.
\end{example}

\begin{figure}[h!]
  \begin{center}
    \includegraphics[width=.9\linewidth]{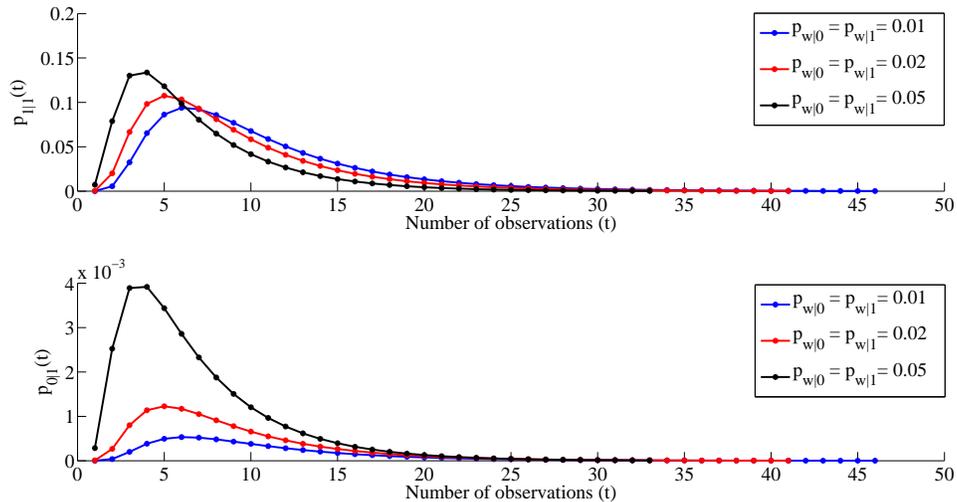}
    \caption{This figure illustrates a typical unimodal set of decision
      probabilities $\{\poo(t)\}_{t\in \N}$ and $\{\pzo(t)\}_{t\in
        \N}$. Here the SDM is implementing the sequential probability ratio
      test with three different accuracy levels (see
      Section~\ref{SSec-sim-nocomm} for more details).}
    \label{fig-pij}
  \end{center}
\end{figure}

\subsection{The $q$ out of $N$ decentralized hypothesis testing}\label{subsec:DHT}

The basic framework for the binary hypothesis testing problem we analyze in
this paper is the one in which there are $N$ SDMs and one fusion center.
The binary hypothesis is denoted by $H$ and it is assumed to take on values
$H_0$ and $H_1$.
Each SDM is assumed to perform individually a binary sequential test;
specifically, for $i\in \until{N}$, at time $t\in \N$, SDM $i$ takes the
observation $x_i(t)$ on a random variable $X_i$, defined on some set
$\mathcal{X}_i$, and it keeps observing $X_i$ until it provides its
decision according to some decision policy $\gamma_i$.  We assume that
\begin{enumerate}
\item the random variables $\{X_i\}_{i=1}^N$ are identical and independent;
\item the SDMs adopt the same decision policy $\gamma$, that is, $\gamma_i\cong\gamma$ for all $i\in \until{N}$;
\item the observations taken, conditioned on either hypothesis, are independent from one SDM to another; 
\item the conditional joint distributions of the individual observations under
each hypothesis are known to the SDMS.
\end{enumerate}
In particular assumptions (i) and (ii) imply that the $N$ decision
processes induced by the $N$ SDMs are all described by the same two sets of
probabilities
\begin{equation}\label{eq:2sets}
  \left\{\pndz\right\} \cup \left\{ \pzz (t), \poz(t)\right\}_{t\in \N}
  \qquad \text{and} \qquad \left\{\pndo\right\}\cup \left\{\pzo(t),
    \poo(t)\right\}_{t\in \N}. 
\end{equation}
We refer to the above property as \emph{homogeneity} among the SDMs. 

Once a SDM arrives to a final local decision, it communicates it to the fusion center. The fusion center collects the messages it receives keeping track of the number of decisions in favor of $H_0$ and in favor of $H_1$. 
A global decision is provided according to a \emph{$q$ out of $N$} counting rule: roughly speaking, as soon as the hypothesis $H_i$ receives $q$ local decisions in its favor, the fusion center globally decides in favor of $H_i$.
In what follows we refer to the above framework as \emph{$q$ out of $N$ sequential decision aggregation} with homogeneous SDMs (denoted as \emph{$q$ out of $N$ SDA}, for simplicity).  

We describe our setup in more formal terms. Let $N$ denote the size of the group of SDMs and let $q$ be a positive integer such that $1\leq q \leq N$, then the \emph{$q$ out of $N$ SDA} with homogeneous SDMs is defined as follows:  

\vspace{0.2cm}
\begin{LaTeXdescription}
\item[SDMs iteration]: 
For each $i \in \until{N}$, the $i$-th SDM keeps observing $X_i$, taking the observations $x_i(1), x_i(2),\dots,$  until time $\tau_i$ where it provides its local decision $d_i \in \{0,1\}$; specifically $d_i=0$ if it decides in favor of $H_0$ and $d_i=1$ if it decides in favor of $H_1$. The decision $d_i$ is instantaneously communicated (i.e., at time $\tau_i$) to the fusion center.
\item[Fusion center state]: The fusion center stores in memory the variables $Count_{0}$ and $Count_{1}$, which are initialized to $0$,
 i.e., $Count_{0}(0)=Count_{1}(0)=0$. If at time $t\in \N$ the fusion center has not yet provided a global decision, then it performs two actions in the following order:\\
 (1) it updates the variables $Count_{0}$ and $Count_{1}$, according to $Count_{0}(t)=Count_{0}(t-1)+n_{0}(t)$ and $Count_{1}(t)=Count_{1}(t-1)+n_{1}(t)$ where $n_{0}(t)$ and $n_{1}(t)$ denote, respectively, the number of decisions equal to $0$ and $1$ received by the fusion center at time $t$.\\
 (2) it checks if one of the following two situations is verified
 \begin{equation}\label{eq:FinDec}
 (i) \left\{
 \begin{array}{l}
 Count_{1}(t)>Count_{0}(t),\\
 Count_{1}(t)\geq q,
 \end{array}
 \right. \qquad
  (ii) \left\{
 \begin{array}{l}
 Count_{1}(t)<Count_{0}(t).\\
 Count_{0}(t)\geq q.
 \end{array}
 \right.
 \end{equation}
 If $(i)$ is verified the fusion center globally decides in favor $H_1$, while if  $(ii)$ is verified the fusion center globally decides in favor of $H_0$.
 Once the fusion center has provided a global decision the \emph{$q$ out of $N$ SDA} algorithm stops.
 \end{LaTeXdescription}   

\begin{remark}[Notes about SDA]
\begin{enumerate}
\item Each SDM has in general a non-zero probability of not giving a
  decision. In this case, the SDM might keep sampling infinitely without
  providing any decision to the fusion center.
\item The fusion center does not need to wait until all the SDM have
  provided a decision before a decision is reach on the group level, as one
  of the two conditions $(i)$ or $(ii)$ in equation~\ref{eq:FinDec} might
  be satisfied much before the $N$ SDM provide their decisions.
\item While we study in this manuscript the case when a fusion center
  receives the information from all SDM, we note that a distributed
  implementation of the SDA algorithm is possible. Analysis similar to the
  one presented here is possible in that case.\oprocend
\end{enumerate}
\end{remark}

\subsection{Problem formulation}\label{subsec:PF}

We introduce now some definitions that will be useful throughout this 
paper. Given a group of $N$ SDMs running the \emph{$q$ out of $N$ SDA} algorithm, $1\leq q \leq N$, we denote
\begin{enumerate}
\item by $T$ the random variable accounting for the number of iterations
  required to provide a decision
  \begin{align*}
    T=&\min \setdef{t}{\text{either \emph{case} (i) or \emph{case} (ii) in
        equation~\eqref{eq:FinDec} is satisfied}};
  \end{align*}
\item  by $\pij(t;N,q)$ the probability of deciding, at time $t$, in favor of $H_i$ given that $H_j$ is correct, i.e.,
\begin{equation}\label{eq:pij}
  \pij(t;N,q):=\P\left[\,\text{Group of $N$ SDMs says } H_i \, | \,H_j, q, T=t\right];
\end{equation}

\item by $\pcj(N,q)$ and $\pwj(N,q)$
the probability of
  correct decision and of wrong decision, respectively, given that $H_j$ is
  the correct hypothesis, i.e.,
\begin{equation}\label{eq:pc/pwN}
\pcj(N,q)=\sum_{t=1}^{\infty}\pjj(t;N,q) \qquad \text{and} 
\qquad \pwj(N,q)=\sum_{t=1}^{\infty}\pij(t;N,q), \,\,i\neq j;
\end{equation}
\item by $\pndj(N,q)$, $j\in \{0,1\}$, the probability of no-decision given that $H_j$ is the correct hypothesis, i.e.,
\begin{equation}\label{eq:noDecN}
  \pndj(N,q):=1-\sum_{t=1}^{\infty}\left(\pzj(t;N,q)+\poj(t;N,q)\right)=1-\pwj(N,q)-\pcj(N,q);
\end{equation}
\item by $\E\left[T|H_j,N,q\right]$ the average number of iterations required by the algorithm to provide a decision, given that $H_j$ is the correct hypothesis, i.e.,
\begin{equation}\label{eq:ETN}
\E\left[T|H_j,N,q\right]:= 
\begin{cases}
  \sum_{t=1}^{\infty}\, t(\pzj(t;N,q)+\poj(t;N,q)), & \text{if } \pndj(N,q)=0,\\
  +\infty, & \text{if }   \pndj(N,q)>0. 
\end{cases}
\end{equation}
\end{enumerate}
Observe that $\pij(t;1,1)$ coincides with the probability $\pij(t)$ introduced in~\eqref{eq:SDMpro}. For ease of notation we will continue using $\pij(t)$ instead of $\pij(t;1,1)$.  

We are now ready to formulate the problem we aim to solve in this paper. 

\begin{problem}[Sequential decision aggregation]
  Consider a group of $N$ homogeneous SDMs with decision probabilities
  $\left\{\pndz\right\} \cup \left\{ \pzz (t), \poz(t)\right\}_{t\in \N}$
  and $\left\{\pndo\right\}\cup \left\{\pzo(t), \poo(t)\right\}_{t\in
    \N}$. Assume the $N$ SDMs run the \emph{$q$ out of $N$ SDA} algorithm
  with the purpose of deciding between the hypothesis $H_0$ and $H_1$. For
  $j\in \{0,1\}$, compute the distributions
  $\left\{\pij(t;N,q)\right\}_{t\in \N}$ as well as the probabilities of
  correct and wrong decision, i.e., $\pcj(N,q)$ and $\pwj(N,q)$, the
  probability of no-decision $\pndj(N,q)$ and the average number of
  iterations required to provide a decision, i.e.,
  $\E\left[T|H_j,N,q\right]$.
\end{problem}

We will focus on the above problem in the next two Sections, both through
theoretical and numerical results. Moreover, in
Section~\ref{sec-asymp-analysis}, we will concentrate on two particular
values of $q$, specifically for $q=1$ and $q=\floor{N/2} +1$,
characterizing the tradeoff between the expected decision time, the
probabilities of correct and wrong decision and the size of the group of
SDMs. When $q=1$ and $q=\lceil N/2 \rceil$, we will refer to the \emph{$q$
  out of $N$ rule} as the \emph{fastest rule} and the \emph{majority rule},
respectively. In this case we will use the following notations
 $$
 \pcjf(N):=\pcj(N;q=1), \qquad  \pwjf(N):=\pwj(N;q=1)
 $$
 and
  $$
 \pcjm(N):=\pcj(N;q=\floor{N/2}+1), \qquad  \pwjm(N):=\pwj(N;q=\floor{N/2}+1).
 $$
 
We end this Section by stating two propositions characterizing the \emph{almost-surely decisions} and \emph{finite expected decision time} properties for the group of SDMs. 

\begin{proposition}\label{pro:Gasd}
  Consider a group of $N$ SDMs running the \emph{$q$ out of $N$ SDA}
  algorithm.  Let the decision-probabilities of each SDM be as
  in~\eqref{eq:2sets}. For $j\in \{0,1\}$, assume there exists at least one
  time instant $t_j\in \N$ such that both probabilities $\pzj(t_j)$ and
  $\poj(t_j)$ are different from zero.  Then the group of SDMs has the
  \emph{almost-sure decision} property if and only if
  \begin{enumerate}
  \item the single SDM has the  \emph{almost-sure decision} property;
  \item $N$ is odd; and
  \item $q$ is such that $1\leq q \leq \lceil N/2 \rceil$.
  \end{enumerate} 
\end{proposition}
\begin{proof}
First we prove that if the group of SDMs has the \emph{almost-sure decision} property, then properties (i), (ii) and (iii) are satisfied. To do so, we show that if one between the properties (i), (ii) and (iii) fails then there exists an event of probability non-zero that leads the group to not provide a decision. First assume that the single SDM does not have the \emph{almost-sure decision} property, i.e., $\pndj>0$, $j\in \{0,1\}$. Clearly this implies that the event \emph{"all the SDMs of the group do not provide a decision"} has probability of occurring equal to $\pndj^N$ which is strictly greater than zero. 
Second assume that $N$ is even and consider the event \emph{"at time $t_j$, $N/2$ SDMs decide in favor of $H_0$ and $N/2$ SDMs decide in favor of $H_1$"}. Simple combinatoric and probabilistic arguments show that the probability of this event is ${N \choose N/2}\,\pzj^{N/2}\,\poj^{N/2}$, which is strictly greater than zero because of the assumption $\pzj(t_j)\neq 0$ and  $\poj(t_j)\neq 0$. Third assume that $q> \lfloor N/2 \rfloor +1$. In this case we consider the event \emph{"at time $t_j$, $\lceil N/2 \rceil$ SDMs decide in favor of $H_0$ and $\lfloor N/2 \rfloor$ SDMs decide in favor of $H_1$"} that, clearly, leads the group of SDMs to not provide a global decision for any $q >\lfloor N/2 \rfloor +1$. Similarly to the previous case, we have that the probability of this event is ${N \choose \lceil N/2 \rceil}\,\pzj^{\lceil N/2 \rceil}\,\poj^{\lfloor N/2 \rfloor}>0$.

We prove now that if properties (i), (ii) and (iii) are satisfied then the
group of SDMs has the \emph{almost-sure decision} property. Observe that,
since each SDM has the \emph{almost-sure decision} property, there exists
almost surely a $N$-tuple $\left(t_1, \ldots, t_N\right) \in \N^N$ such
that the $i$-th SDM provides its decision at time $t_i$. Let
$\bar{t}:=\max\setdef{t_i}{i\in \until{N}}$. Since $N$ is odd, then
$Count_{1}(\bar{t})\neq Count_{0}(\bar{t})$. Moreover since $q\leq
\floor{N/2}+1$ and $Count_{1}(\bar{t})+Count_{0}(\bar{t})=N$, either
$Count_{1}(\bar{t})\geq q$ or $Count_{0}(\bar{t})\geq q$ holds true. Hence
the fusion center will provide a global decision not later than time
$\bar{t}$.
\end{proof}

\begin{proposition}\label{pro:Gfedt}
  Consider a group of $N$ SDMs running the \emph{$q$ out of $N$ SDA}
  algorithm.  Let the decision-probabilities of each SDM be as
  in~\eqref{eq:2sets}. For $j\in \{0,1\}$, assume there exists at least one
  time instant $t_j\in \N$ such that both probabilities $\poj(t_j)$ and
  $\poj(t_j)$ are different from zero. Then the group of SDMs has the
  \emph{finite expected decision time} property if and only if
  \begin{enumerate}
  \item the single SDM has the  \emph{finite expected decision time} property;
  \item $N$ is odd; and
  \item $q$ is such that $1\leq{q}\leq\ceil{N/2}$.
  \end{enumerate} 
\end{proposition}
\begin{proof}
  The proof follows the lines of the proof of the previous proposition.
\end{proof}

\begin{remark}\label{rem:ASD}
  The existence, for $j\in \{0,1\}$, of a time $t_j$ such that
  $\pzj(t_j)\neq 0$ and $\poj(t_j)\neq 0$, is necessary only for proving
  the ''if'' side of the previous propositions. In other words the validity
  of properties (i), (ii) and (iii) in Proposition~\ref{pro:Gasd} (resp. in
  Prop.~\ref{pro:Gfedt}) guarantees that the group of SDMs possesses the
  \emph{almost-sure decision} property (resp. the \emph{finite expected
    decision time} property.)  \oprocend
\end{remark}

\section{Recursive analysis of the $q$-out-of-$N$ sequential aggregation rule} \label{Sec-analysis-nocomm}

The goal of this section is to provide an efficient method to compute the
probabilities $\pij(t;N,q)$, $i,j\in\{0,1\}$. These probabilities, using
equations~\eqref{eq:pc/pwN},~\eqref{eq:noDecN} and~\eqref{eq:ETN} will
allow us to estimate the probabilities of correct decision, wrong decision
and no-decision, as well as the expected number of iterations required to
provide the final decision.

We first consider in subsection~\ref{subsec:1-q-N/2} the case where $1\leq
q \leq \lfloor N/2 \rfloor$; in subsection~\ref{subsec:N/2-q-N} we consider
the case where $\lfloor N/2 \rfloor+1 \leq q \leq N$.

\subsection{Case $1\leq q \leq \lfloor N/2 \rfloor$}\label{subsec:1-q-N/2}

To present our analysis method, we begin with an informal description of
the decision events characterizing the \emph{$q$ out of $N$} SDA algorithm.
Assume that the fusion center provides its decision at time~$t$. This fact
implies that neither case (i) nor case (ii) in equation~\eqref{eq:FinDec}
has happened at any time before $t$. Moreover, two distinct set of events
may precede time $t$, depending upon whether the values of the counters
$Count_0$ and $Count_1$ at time $t-1$ are smaller than $q$ or not.
In a first possible set of events, say the ``simple situation,'' the
counters satisfy $0\leq Count_0(t-1), Count_1(t-1)\leq q-1$ and, hence, the
time $t$ is the first time that at least one of the two counters crosses
the threshold $q$.
In a second possible set of events, say the ``canceling situation,'' the
counters $Count_0(t-1)$ and $Count_1(t-1)$ are greater than $q$ and,
therefore, equal. In the canceling situation, there must exist a time
instant $\bar{\tau}\leq t-1$ such that $Count_{0}(\bar{\tau}-1)<q$,
$Count_{1}(\bar{\tau}-1)<q$ and $Count_{0}(\tau)=Count_{1}(\tau)\geq q$ for
all $\tau \in \{\bar{\tau}+1,\dots,t-1\}$. In other words, both counters
cross the threshold $q$ at the same time instant $\bar{\tau}$ reaching the
same value, that is, $Count_0(\bar{\tau})=Count_1(\bar{\tau})$, and, for
time $\tau \in \{\bar{\tau}+1,\dots,t-1\}$, the number $n_0(\tau)$ of SDMs
deciding in favor of $H_0$ at time $\tau$ and the number $n_1(\tau)$ of
SDMs deciding in favor of $H_1$ at time $\tau$ cancel each other out, that
is, $n_0(\tau)=n_1(\tau)$.  

In what follows we study the probability of the simple and canceling
situations.  To keep track of both possible set of events, we introduce
four probability functions, $\alpha$, $\beta$, $\bar{\alpha}$,
$\bar{\beta}$. The functions $\alpha$ and $\beta$ characterize the simple
situation, while $\bar{\alpha}$ and $\bar{\beta}$ characterize the
canceling situation.  First, for the simple situation, define the
probability function $\alpha:\N \times \left\{0,\ldots,q-1 \right\}\times
\left\{0,\ldots,q-1\right\}\to [0,1]$ as follows: given a group of
$s_0+s_1$ SDMs, $\alpha(t, s_0, s_1)$ is the probability that
\begin{enumerate}
\item all the $s_0+s_1$ SDMs have provided a decision up to time $t$; and
\item considering the variables $Count_{0}$ and $Count_{1}$ restricted to
  this group of $s_0+s_1$ SDMs , $Count_{0}(t)=s_0$ and $Count_{1}(t)=s_1$.
\end{enumerate}
Also, define the probability function $\map{{\beta}_{1|j}}{\natural\times
  \left\{0,\ldots, q-1 \right\}\times \left\{0,\ldots,
    q-1\right\}}{[0,1]}$, $j\in\{0,1\}$ as follows: given a group of
$N-(s_0+s_1)$ SDMs, $\beta_{1|j}(t,s_0,s_1)$ is the probability that
\begin{enumerate}
\item no SDMs have provided a decision up to time $t-1$; and
\item considering the variables $Count_{0}$ and $Count_{1}$ restricted to
  this group of $N-(s_0+s_1)$ SDMs, $Count_{0}(t)+s_0< Count_{1}(t)+s_1$,
  and $Count_{1}(t)+s_1 \geq q$.
\end{enumerate}
Similarly, it is straightforward to define the probabilities
${\beta}_{0|j}$, $j\in\{0,1\}$.

Second, for the canceling situation, define the probability function
$\bar{\alpha} : \N \times \left\{q, \ldots, \lfloor N/2 \rfloor\right\}\to
[0,1]$ as follows: given a group of $2s$ SDMs, $\bar{\alpha}(t,s)$ is the
probability that
\begin{enumerate}
\item all the $2s$ SDMs have provided a decision up to time $t$; and
\item there exists $\bar{\tau}\leq t$ such that, considering the variables
  $Count_{0}$ and $Count_{1}$ restricted to this group of $2s$ SDMs
  \begin{itemize}
  \item $Count_{0}(\bar{\tau}-1)<q$ and $Count_{1}(\bar{\tau}-1)<q$;
  \item $Count_{0}(\tau)=Count_{1}(\tau)\geq q$ for all $\tau \geq
    \bar{\tau}$.
  \end{itemize}
\end{enumerate}
Also, define the probability function
$\map{{\bar{\beta}}_{1|j}}{\natural\times \left\{q,\ldots \lfloor N/2
    \rfloor \right\}}{[0,1]}$, $j\in\{0,1\}$ as follows: given a group of
$N-2s$ SDMs, $\bar{\beta}_{1|j}(t,s)$ is the probability that
\begin{enumerate}
\item no SDMs have provided a decision up to time $t-1$; and
\item at time $t$ the number of SDMs providing a decision in favor of $H_1$
  is strictly greater of the number of SDMs providing a decision in favor
  of $H_0$.
\end{enumerate}
Similarly, it is straightforward to define the probabilities
${\bar{\beta}}_{0|j}$, $j\in\{0,1\}$.

Note that, for simplicity, we do not explicitly keep track of the
dependence of the probabilities $\beta$ and $\bar\beta$ upon the numbers
$N$ and $q$.  The following proposition shows how to compute the
probabilities $\left\{\pij(t;N,q)\right\}_{t=1}^{\infty}$,
$i,j\in\{0,1\}$, starting from the above definitions.
 
\begin{proposition}\mbox{} \textbf{(\emph{$q$ out of $N$}: a recursive
    formula)}
  \label{Prop-Numq}
  Consider a group of $N$ SDMs, running the \emph{$q$ out of $N$ SDA}
  algorithm. Without loss of generality, assume $H_1$ is the correct
  hypothesis. Then, for $i\in \{0,1\}$, we have, for $t=1$,
  \begin{align}\label{eq:t11}
    \pio(1;N,q)=\beta_{i|1}(1,0,0),
  \end{align}   
  and, for $t\geq 2$, 
  \begin{align} \label{eq-Main} & \pio(t;N,q) = \sum_{s_0=0}^{q-1}
    \sum_{s_1=0}^{q-1} {N \choose s_1+s_0} \alpha(t-1, s_0, s_1)
    \beta_{i|1}(t, s_0, s_1)+ \sum_{s=q}^{\lfloor N/2\rfloor} {N \choose
      2s} \bar{\alpha}(t-1, s) \bar{\beta}_{i|1}(t, s).
  \end{align}
 \end{proposition}

 \begin{proof}
   The proof that formulas in~\eqref{eq:t11} hold true follows trivially
   form the definition of the quantities $\beta_{1|1}(1,0,0)$ and
   $\beta_{0|1}(1,0,0)$. We start by providing three useful definitions.
 
   First, let $E_t$ denote the event that the SDA with the \emph{$q$ out of
     $N$} rule provides its decision at time $t$ in favor of $H_1$.
 
 Second, for $s_0$ and $s_1$ such that $0\leq s_0,s_1 \leq q-1$, let $E_{s_0,s_1,t}$ denote the event such that 
 \begin{enumerate}
 \item there are $s_0$ SDMs that have decided in favor of $H_0$ up to time $t-1$;
 \item  there are $s_1$ SDMs that have decided in favor of $H_1$ up to time $t-1$;
 \item there exist two positive integer number $r_0$ and $r_1$ such that  
 \begin{itemize}
 \item $s_0+r_0< s_1+r_1$ and $s_1+r_1\geq q$.
 \item at time $t$, $r_0$ SDMs decides in favor of $H_0$ while $r_1$ SDMs decides in favor of $H_1$ 
 \end{itemize}
 \end{enumerate}
 Third, for $q\leq s \leq \lfloor N/2 \rfloor$, let $E_{s,t}$ denote the event such that
 \begin{enumerate}
 \item $2s$ SDMs have provided their decision up to time $t-1$ balancing their decision, i.e., there exists $\bar{\tau}\leq t-1$ with the properties that, considering the variables $Count_{-}$ and $Count_{+}$ restricted to these $2s$ SDMs 
 \begin{itemize}
 \item $Count_{0}(\tau)< q$, $Count_{1}(\tau)<q$, for $1\leq \tau \leq \bar{\tau}-1$;
 \item $Count_{0}(\tau)=Count_{1}(\tau)$ for $\bar{\tau} \leq \tau \leq t-1$;
 \item $Count_{0}(t-1)=Count_{1}(t-1)=s$.
 \end{itemize}
 \item at time $t$ the number of SDMs providing their decision in favor of $H_1$ is strictly greater than the number of SDMs deciding in favor of $H_0$.
 \end{enumerate}
 Observe that 
 $$
 E_t=\left(\mathop{\cup}_{0\leq s_0, s_1\leq q-1} E_{s_0,s_1,t}\right) \bigcup \left(\mathop{\cup}_{q\leq s \leq \lfloor N/2 \rfloor} E_{s,t}\right).
 $$
 Since $E_{s_0,s_1,t}$, $0\leq s_0, s_1\leq q-1$, and $E_{s,t}$, $q \leq
 s\leq \lfloor N/2 \rfloor$ are disjoint sets, we can write
 \begin{equation}\label{eq:Tot}
 \P\left[E_t\right]= \mathop{\sum}_{0\leq s_0,s_1\leq q-1} \P\left[E_{s_0,s_1,t}\right]+ \mathop{\sum}_{q\leq s \leq \lfloor N/2 \rfloor}\P\left[E_{s,t}\right]. 
 \end{equation}
 Observe that, according to the definitions of $\alpha(t-1, s_0, s_1)$, $\bar{\alpha}(t-1, s)$, $ \beta_{1|1}(t, s_0, s_1)$ and $\bar{\beta}_{1|1}(t, s)$, provided above,
 \begin{equation}\label{eq:Par1}
  \P\left[E_{s_0,s_1,t}\right]={N \choose s_1+s_0} 
    \alpha(t-1, s_0, s_1) \beta_{1|1}(t, s_0, s_1)
 \end{equation}
 and that
  \begin{equation}\label{eq:Par2}
 \P\left[E_{s,t}\right]={N \choose 2s} \bar{\alpha}(t-1, s) \bar{\beta}_{1|1}(t, s).
 \end{equation}
 Plugging equations~\eqref{eq:Par1} and~\eqref{eq:Par2} into equation~\eqref{eq:Tot} concludes the proof of the Theorem. 
 \end{proof}
 Formulas, similar to the ones in~\eqref{eq:t11} and~\eqref{eq-Main} can be
 provided for computing also the probabilities
 $\left\{\piz(t;N,q)\right\}_{t=1}^{\infty}$, $i\in\{0,1\}$.
 
As far as the probabilities $\alpha(t,s_0,s_1)$, $\bar{\alpha}(t,s)$,
$\beta_{i|j}(t,s_0,s_1)$, $\bar{\beta}_{i|j}(t,s)$, $i,j\in\{0,1\}$, are
concerned, we now provide expressions to calculate them.

\begin{proposition}\label{prop:AlphaBeta}
  Consider a group of $N$ SDMs, running the \emph{$q$ out of $N$ SDA}
  algorithm for $1\leq q \leq \lfloor N/2 \rfloor$. Without loss of
  generality, assume $H_1$ is the correct hypothesis. For $i\in \{0,1\}$,
  let $\map{\pi_{i|1}}{\natural}{[0,1]}$ denote the cumulative probability
  up to time $t$ that a single SDM provides the decision $H_i$, given that
  $H_1$ is the correct hypothesis, i.e.,
  \begin{equation}\label{eq:gamma}
    \pi_{i|1}(t)=\sum_{s=1}^{t}\pio(t).
      \end{equation} 
  For $t\in\natural$, $s_0,s_1\in\until{q-1}$, $s\in\left\{q, \ldots,
    \floor{N/2}\right\}$, the probabilities $\alpha(t,s_0,s_1)$,
  $\bar{\alpha}(t,s)$, $\beta_{1|1}(t,s_0,s_1)$, and
  $\bar{\beta}_{1|1}(t,s)$ satisfy the following relationships (explicit
  for $\alpha$, $\beta$, $\bar{\beta}$ and recursive for $\bar{\alpha}$):
  \begin{align*}
    \alpha(t,s_0,s_1) &= {s_0+s_1 \choose s_0} \pi_{0|1}^{s_0}(t)
    \pi_{1|1}^{s_1}(t),
    \\
    \bar{\alpha}(t,s) &= \sum_{s_0=0}^{q-1} \sum_{s_1=0}^{q-1} {2s \choose
      s_0+s_1} {2s-s_0-s_1 \choose s-s_0}\alpha(t-1,s_0,s_1)
    \pzo^{s-s_0}(t) \poo^{s-s_1}(t) \\
    &\qquad\qquad\qquad\qquad\qquad\qquad\qquad+\sum_{h=q}^{s} {2s \choose
      2h} {2s-2h \choose s-h}\bar{\alpha}(t-1, h) \pzo^{s-h}(t)
    \poo^{s-h}(t), 
    \\
    \beta_{1|1}(t,s_0,s_1) &= \sum_{h_1=q-s_1}^{N-\bar{s}} {N-\bar{s}
      \choose h_1} \poo^{h_1}(t) \left[\sum_{h_0=0}^{m}{N-\bar{s}-h_1
        \choose h_0} \pzo^{h_0}(t)\left(1- \pi_{1|1}(t)-\pi_{0|1}(t)
      \right)^{N-\bar{s}-h_0-h_1}\right],
    \\
    \bar{\beta}_{1|1}(t,s) &= \sum_{h_1=1}^{N-2s} {N-2s \choose h_1}
    \poo^{h_1}(t) \left[\sum_{h_0=0}^{\bar{m}} {N-2s-h_1\choose h_0 }
      \pzo^{h_0}(t) (1- \pi_{1|1}(t)-\pi_{0|1}(t))^{N-2s-h_0-h_1}\right],
  \end{align*}
  where $\bar{s}=s_0+s_1$, $m=\min \{h_1+s_1-s_0-1,N-(s_0+s_1)-h_1 \}$ and
  $\bar{m}=\min \{h_1-1,N-2s-h_1 \}$.  Moreover, corresponding relationships
  for $\beta_{0|1}(t,s_0,s_1)$ and $\bar{\beta}_{0|1}(t,s)$ are obtained by
  exchanging the roles of $\poo(t)$ with $\pzo(t)$ in the relationships for
  $\beta_{1|1}(t,s_0,s_1)$ and $\bar{\beta}_{1|1}(t,s)$.
\end{proposition}

\begin{proof}
  The evaluation of $\alpha(t,s_0,s_1)$ follows from standard probabilistic
  arguments. Indeed, observe that, given a first group of $s_0$ SDMs and a
  second group of $s_1$ SDMs, the probability that all the SDMs of the
  first group have decided in favor of $H_0$ up to time $t$ and all the
  SDMs of the second group have decided in favor of $H_1$ up to time $t$ is
  given by $\pi_{0|1}^{s_0}(t)\pi_{1|1}^{s_1}(t)$. The desired result follows
  from the fact that there are ${s_1+s_0 \choose s_0}$ ways of dividing a
  group of $s_0+s_1$ SDMs into two subgroups of $s_0$ and $s_1$ SDMs.

  Consider now $\bar{\alpha}(t,s)$. Let $E_{\bar{\alpha}(t,s)}$ denote the event
  of which $\bar{\alpha}(t,s)$ is the probability of occurring, that is, the event that, given a group of $2s$
  SDMs,
 \begin{enumerate}
 \item all the $2s$ SDMs have provided a decision up to time $t$; and
 \item there exists $\bar{\tau}\leq t$ such that, considering the variables
   $Count_{0}$ and $Count_{1}$ restricted to this group of $2s$ SDMs
   \begin{itemize}
   \item $Count_{0}(\bar{\tau}-1)<q$ and $Count_{1}(\bar{\tau}-1)<q$;
   \item $Count_{0}(\tau)=Count_{1}(\tau)\geq q$ for all $\tau \geq
     \bar{\tau}$.
   \end{itemize}
 \end{enumerate}
 Now, for a group of $2s$ SDMs, for $0\leq s_0,s_1 \leq q-1$, let
 $E_{t-1,s_0,s_1}$ denote the event that
 \begin{enumerate}
 \item $s_0$ (resp. $s_1$) SDMs have decided in favor of $H_0$
   (resp. $H_1$) up to time $t-1$;
 \item $s-s_0$ (resp. $s-s_1$) SDMs decide in favor of $H_0$ (resp. $H_1$)
   at time $t$.
 \end{enumerate}
 Observing that for $s_0+s_1$ assigned SDMs the probability that fact (i) is verified is given by $\alpha(t-1,s_0,s_1)$ we can write that
 $$
 \P[E_{t-1,s_0,s_1}]={2s \choose s_0+s_1} {2s-s_0-s_1 \choose
   s-s_0}\alpha(t-1,s_0,s_1) \pzo^{s-s_0}(t) \poo^{s-s_1}(t).
 $$
 Consider again a group of $2s$ SDMs and for $q\leq h \leq s$ let
 $\bar{E}_{t-1,h}$ denote the event that
 \begin{enumerate}
 \item  $2h$ SDMs have provided a decision up to time $t-1$; 
 \item there exists $\bar{\tau}\leq t-1$ such that, considering the
   variables $Count_{0}$ and $Count_{1}$ restricted to the group of $2h$
   SDMs that have already provided a decision,
   \begin{itemize}
   \item $Count_{0}(\bar{\tau}-1)<q$ and $Count_{1}(\bar{\tau}-1)<q$;
   \item $Count_{0}(\tau)=Count_{1}(\tau)\geq q$ for all $\tau \geq
     \bar{\tau}$; and
   \item $Count_{0}(t-1)=Count_{1}(t-1)=h$;
   \end{itemize}
 \item at time instant $t$, $s-h$ SDMs decide in favor of $H_0$ and $s-h$
   SDMs decide in favor of $H_1$.
 \end{enumerate}
 Observing that for $2h$ assigned SDMs the probability that fact (i) and fact (ii) are verified is given by $\bar{\alpha}(t-1, h)$, we can write that
 $$
 \P[\bar{E}_{t-1,h}]={2s \choose
      2h} {2s-2h \choose s-h}\bar{\alpha}(t-1, h) \pzo^{s-h}(t)
    \poo^{s-h}(t).
 $$
 Observe that
$$
E_{\bar{\alpha}(t,s)}= \left(\bigcup_{s_0=0}^{q} \bigcup_{s_1=0}^q E_{t-1,s_0,s_1}
\right) \bigcup \left(\bigcup_{h=q}^{\floor{N/2}} \bar{E}_{t-1,h}\right).
$$
Since the events $E_{t-1,s_0,s_1}$, $0\leq s_0, s_1 < q$ and $
\bar{E}_{t-1,h}$, $ q\leq h\leq \floor{N/2}$, are all disjoint we have that
$$ 
\P[E_{\bar{\alpha}(t,s)}] = \sum_{s_0=0}^{q-1} \sum_{s_1=0}^{q-1} \P[E_{t-1,s_0,s_1}]+\sum_{h=q}^{s} \P[\bar{E}_{t-1,h}].
$$
Plugging the expressions of $\P[E_{t-1,s_0,s_1}]$ and $\P[\bar{E}_{t-1,h}]$
in the above equality gives the recursive relationship for computing
$\bar{\alpha}(t,s)$.
   
Consider now the probability $\beta_{1|1}(t,s_0,s_1)$. Recall that this
probability refers to a group of $N-(s_0+s_1)$ SDMs. Let us introduce some
notations. Let $E_{\beta_{1|1}(t,s_0,s_1)}$ denote the event of which
$\beta_{1|1}(t,s_0,s_1)$ represents the probability of occurring and let
$E_{t;h_1,s_1, h_0,s_0}$ denote the event that, at time $t$
\begin{itemize}
\item $h_1$ SDMs decides in favor of $H_1$;
\item $h_0$ SDMs decides in favor of $H_0$;
\item the remaining $N-(s_0+s_1)-(h_0+h_1)$ do not provide a decision up to time $t$.
\end{itemize}
Observe that the above event is well-defined if and only if $h_0+h_1 \leq
N-(s_0+s_1)$. Moreover $E_{t;h_1,s_1, h_0,s_0}$ contributes to
$\beta_{1|1}(t,s_0,s_1)$, i.e., $E_{t;h_1,s_1, h_0,s_0} \subseteq
E_{\beta_{1|1}(t,s_0,s_1)}$ if and only if $h_1\geq q-s_1$ and
$h_0<h_1+s_1-s_0$ (the necessity of these two inequalities follows directly from the definition of $\beta_{1|1}(t,s_0,s_1)$). Considering the three inequalities $h_0+h_1 \leq
N-(s_0+s_1)$, $h_1\geq q-s_1$ and $h_0<h_1+s_1-s_0$, it follows
that
$$
E_{\beta_{1|1}(t,s_0,s_1)}= \bigcup \big\{E_{t;h_1,s_1, h_0,s_0} \;|\;
q-s_1 \leq h_1\leq N-(s_0+s_1) \,\,\,\,\text{and}\,\,\,\, h_0\leq m \big\},
$$
where $m=\min \{h_1+s_1-s_0-1,N-(s_0+s_1)-h_1 \}$. To conclude it suffices
to observe that the events $E_{t;h_1,s_1, h_0,s_0}$ for $q-s_1 \leq h_1\leq
N-(s_0+s_1)$ and $h_0\leq m$ are disjoint events and that
$$
\P[E_{t;h_1,s_1, h_0,s_0}]={N-\bar{s} \choose
  j} \poo^{h_1}(t) {N-\bar{s}-h_1 \choose h_0}
\pzo^{h_0}(t)\left(1- \pi_{1|1}(t)-\pi_{0|1}(t)
\right)^{N-\bar{s}-h_0-h_1},
$$
where $\bar{s}=s_0+s_1$.

The probability $\bar{\beta}_{1|1}(t,s)$ can be computed reasoning
similarly to $\beta_{1|1}(t,s_0,s_1)$.
\end{proof}

Now we describe some properties of the above expressions in order to assess
the computational complexity required by the formulas introduced in
Proposition \ref{Prop-Numq} in order to compute
$\left\{\pij(t;N,q)\right\}_{t=1}^{\infty}$, $i,j\in\{0,1\}$. From the
expressions in Proposition~\ref{prop:AlphaBeta}
we observe that
 \begin{itemize}
 \item $\alpha(t,s_0,s_1)$ is a function of $\pi_{0|1}(t)$ and
   $\pi_{1|1}(t)$;
    \item $\bar{\alpha}(t,s)$ is a function of $\alpha(t-1,s_0,s_1)$, $0\leq
   s_0,s_1\leq q-1$, $\pzo(t)$, $\poo(t)$ and $\bar{\alpha}(t-1,h)$, $q\leq
   h \leq s$;
 \item $\beta_{i|1}(t,s_0,s_1)$, $\bar{\beta}_{i|1}$, $i\in\{0,1\}$, are
   functions of $\pzo(t)$, $\poo(t)$, $\pi_{0|1}(t)$ and
   $\pi_{1|1}(t)$.
 \end{itemize}
 Moreover from equation~\eqref{eq:gamma} we have that $\pi_{i|j}(t)$ is
 a function of $\pi_{i|j}(t-1)$ and $\pij(t)$.

 Based on the above observations, we deduce that $\pzo(t;N,q)$ and
 $\poo(t;N,q)$ can be seen as the output of a dynamical system having the
 $(\lfloor N/2 \rfloor-q+3)$-th dimensional vector with components the
 variables $\pi_{0|1}(t-1)$, $\pi_{1|1}(t-1)$, $\bar{\alpha}(t-1,s)$,
 $q\leq h \leq \lfloor N/2 \rfloor$ as states and the two dimensional
 vector with components $\pzo(t)$, $\poo(t)$, as inputs. As a consequence,
 it follows that the iterative method we propose to compute
 $\left\{\pij(t;N,q)\right\}_{t=1}^{\infty}$, $i,j\in\{0,1\}$, requires
 keeping in memory a number of variables which grows linearly with the
 number of SDMs.

\subsection{Case $\lfloor N/2 \rfloor +1 \leq q \leq N$}\label{subsec:N/2-q-N}

The probabilities $\pij(t;N,q)$, $i,j\in\{0,1\}$ in the case where $\lfloor N/2 \rfloor +1 \leq q \leq N$ can be computed according to the expressions reported in the following Proposition. 

\begin{proposition}
  Consider a group of $N$ SDMs, running the \emph{$q$ out of $N$ SDA}
  algorithm for $\lfloor N/2 \rfloor+1 \leq q \leq N$. Without loss of
  generality, assume $H_1$ is the correct hypothesis.  For $i\in \{0,1\}$,
  let $\map{\pi_{i|1}}{\natural}{[0,1]}$ be defined
  as~\eqref{eq:gamma}. Then, for $i\in \{0,1\}$, we have for $t=1$
\begin{align}\label{eq:q>N/2t=1}
\pio(1;N,q)=\sum_{h=q}^{N}{N \choose h}\pio^h(1)\left(1-\pio(1)\right)^{N-h}
\end{align}
and for $t\geq 2$
 \begin{align}\label{eq:q>N/2}
\pio(t;N,q) =& \sum_{k=0}^{q-1}{N \choose k} \pi^k_{i|1}(t-1) \sum_{h=q-k}^{N-k}  {N-k \choose h} \pio^h(t) \left(1-\pi_{i|1}(t) \right)^{N-(h+k)}.
\end{align}
\end{proposition}
\begin{proof}
Let $t=1$. Since $q>N/2$, the probability that the fusion center decides in favor of $H_i$ at time $t=1$ is given by the probability that al least $q$ SDMs decide in favor of $H_i$ at time $1$. From standard combinatoric arguments this probability is given by~\eqref{eq:q>N/2t=1}. 

If $t>1$, the probability that the fusion center decides in favor of $H_i$ at time $t$ is given by the probability that $h$ SDMs, $0\leq h< q$, have decided in favor of $H_i$ up to time $t-1$, and that at least $q-h$ SDMs decide in favor of $H_i$ at time $t$. Formally let $E_t^{(i)}$ denote the event that the fusion center provides its decision in favor of $H_i$ at time $t$ and let $E_{h,t; k,t-1}^{(i)}$ denote the event that $k$ SDMs have decided in favor of $H_i$ up to time $t-1$ and $h$ SDMs decide in favor of $H_i$ at time $t$. 
Observe that
$$
E_t^{(i)}=\bigcup_{k=0}^{q-1}\,\, \bigcup_{h=q-k}^{N-k} E_{h,t; k,t-1}^{(i)}.
$$
Since $E_{h,t; k,t-1}^{(i)}$ are disjoint sets it follows that
$$
\P\left[E_t^{(i)}\right]=\sum_{k=0}^{q-1}\,\, \sum_{h=q-k}^{N-k} \P\left[E_{h,t; k,t-1}^{(i)}\right].
$$
The proof is concluded by observing that 
$$
\P\left[E_{h,t; k,t-1}^{(i)}\right]={N \choose k} \pi^k_{i|1}(t-1) {N-k \choose h} \pio^h(t) \left(1-\pi_{i|1}(t) \right)^{N-(h+k)}.
$$
\end{proof}

Regarding the complexity of the expressions in~\eqref{eq:q>N/2} it is easy
to see that the probabilities $\pij(t;N,q)$, $i,j\in\{0,1\}$ can be
computed as the output of a dynamical system having the two dimensional
vector with components $\pi_{0|1}(t-1), \pi_{1|1}(t-1)$ as state and the
two dimensional vector with components $\pzo(t), \poo(t)$ as input. In this
case the dimension of the system describing the evolution of the desired
probabilities is independent of $N$.
  
\section{Scalability analysis of the fastest and majority sequential
  aggregation rules}\label{sec-asymp-analysis}
 
The goal of this section is to provide some theoretical results
characterizing the probabilities of being correct and wrong for a group
implementing the \emph{q-out-of-N} SDA rule. We also aim to characterize
the probability with which such a group fails to reach a decision in
addition to the time it takes for this group to stop running any test.  In
Sections~\ref{sec:FixedQ-fastest} and~\ref{sec:FixedQ-majority} we consider
the fastest and the majority rules, namely the thresholds $q=1$ and
$q=\lceil N/2 \rceil$, respectively; we analyze how these two counting
rules behave for increasing values of $N$.  In Section~\ref{sec:FixedN}, we
study how these quantities vary with arbitrary values $q$ and fixed values
of $N$.

\subsection{The fastest rule for varying values of  $N$}\label{sec:FixedQ-fastest}
In this section we provide interesting characterizations of accuracy and
expected time under the \emph{fastest} rule, i.e., the counting rules with
threshold $q=1$. For simplicity we restrict to the case where the group has
the \emph{almost-sure} decision property. In particular we assume the
following two properties.
\begin{assumption}\label{ass:1}
  The number $N$ of SDMs is odd and the SDMs satisfy the \emph{almost-sure}
  decision property.
\end{assumption} 
Here is the main result of this subsection. Recall that
$\pwof(N)$ is the probability of wrong decision by a group of
$N$ SDMs implementing the fastest rule (assuming $H_1$ is the correct
hypothesis).

\begin{proposition}[Accuracy and expected time under the fastest rule]
  \label{prop:A-ET-Fastest}
  Consider the \emph{$q$ out of $N$ SDA} algorithm under
  Assumption~\ref{ass:1}. Assume $q=1$, that is, adopt the \emph{fastest}
  SDA rule.  Without loss of generality, assume $H_1$ is the correct
  hypothesis. Define the \emph{earliest possible decision time}
  \begin{equation}\label{eq:bart}
    \bar{t}:=\min \setdef{t\in\N}{\text{either }\poo(t)\neq 0 \text{ or }
      \pzo(t)\neq 0 }.  
  \end{equation}
  Then the probability of error satisfies
  \begin{equation}\label{eq:AccFast}
    \lim_{N\to \infty}\, \pwof(N) = 
    \begin{cases}
      0, \quad & \text{if } \poo(\bar{t})>\pzo(\bar{t}), \\
      1, \quad & \text{if } \poo(\bar{t})<\pzo(\bar{t}), \\
      \frac{1}{2}, \quad & \text{if } \poo(\bar{t})=\pzo(\bar{t}),
    \end{cases}
  \end{equation}
  and the expected decision time satisfies
  \begin{equation}\label{eq:ExpFastest}
    \lim_{N \to \infty} \E\left[T| H_1,N, q=1\right] = \bar{t}.
  \end{equation}
\end{proposition}

\begin{proof}
  We start by observing that in the case where the fastest rule is applied,
  formulas in~\eqref{eq-Main} simplifies to
$$
\poo(t;N, q=1)=\beta_{1|1}(t,0,0), \qquad \text{for all } t\in \N.
$$
Now, since $\poo(t)=\pzo(t)=0$ for $t < \bar{t}$, it follows that 
$$
\poo(t;N,q=1)=\beta_{1|1}(t,0,0)=0,  \qquad t < \bar{t}.
$$
Moreover we have $\pi_{1|1}(\bar{t})=\poo(\bar{t})$ and
$\pi_{0|1}(\bar{t})=\pzo(\bar{t})$. According to the definition of the
probability $\beta_{1|1}(\bar{t}, 0,0)$, we write
\begin{equation*}
  \beta_{1|1}(\bar{t}, 0, 0) =\sum_{j=1}^{N} {N\choose j}  \poo^j(\bar{t})
  \left\{\sum_{i=0}^{m}{N-j \choose i} \pzo^i(\bar{t})\left(1-
      \poo(\bar{t})-\pzo(\bar{t}) \right)^{N-i-j}\right\}, 
\end{equation*}
where $m=\min\left\{j-1, N-j\right\}$, or equivalently
\begin{align}\label{eq:beta11}
  \beta_{1|1}(\bar{t}, 0, 0) &=\sum_{j=1}^{\lfloor N/2 \rfloor} {N\choose
    j}  \poo^j(\bar{t}) \left\{\sum_{i=0}^{j-1}{N-j \choose i}
    \pzo^i(\bar{t})\left(1- \poo(\bar{t})-\pzo(\bar{t})
    \right)^{N-i-j}\right\}\nonumber\\ 
  &\qquad\qquad+ \sum_{j=\lceil N/2 \rceil}^{N} {N\choose j}  \poo^j(\bar{t}) \left\{\sum_{i=0}^{N-j}{N-j \choose i} \pzo^i(\bar{t})\left(1- \poo(\bar{t})-\pzo(\bar{t}) \right)^{N-i-j}\right\}\nonumber\\
  &=\sum_{j=1}^{\lfloor N/2 \rfloor} {N\choose j}  \poo^j(\bar{t}) \left\{\sum_{i=0}^{j-1}{N-j \choose i} \pzo^i(\bar{t})\left(1- \poo(\bar{t})-\pzo(\bar{t}) \right)^{N-i-j}\right\}\nonumber\\
 &\qquad\qquad+ \sum_{j=\lceil N/2 \rceil}^{N} {N\choose j}  \poo^j(\bar{t})\left(1- \poo(\bar{t}) \right)^{N-j}.
\end{align}
An analogous expression for $\beta_{0|1}(\bar{t}, 0, 0)$ can be obtained by
exchanging the roles of $\pzo(\bar{t})$ and $\pzo(\bar{t})$ in
equation~\eqref{eq:beta11}. The rest of the proof is articulated as
follows. First, we prove that
\begin{equation}\label{eq:P1}
  \lim_{N\to \infty} \left(\poo(\bar{t}; N, q=1)+\pzo(\bar{t}; N, q=1)\right)=  \lim_{N\to \infty} \left(\beta_{1|1}(\bar{t}, 0, 0)+\beta_{0|1}(\bar{t}, 0, 0)\right)=1.
\end{equation}
This fact implies that equation~\eqref{eq:ExpFastest} holds and that, if
$\poo(\bar{t})=\pzo(\bar{t})$, then $\lim_{N\to \infty}\,
\pwof(N) = 1/2$. Indeed
\begin{align*}
  \lim_{N\to \infty}\E\left[T|H_j,N,q=1\right]&= \lim_{N\to \infty}
  \sum_{t=1}^{\infty}\, t(\pzj(t;N,q=1)+\pij(t;N,q=1))=\bar{t}.
\end{align*}
Moreover, if $\poo(\bar{t})=\pzo(\bar{t})$, then also
$(\beta_{1|1}(\bar{t}, 0, 0)=\beta_{0|1}(\bar{t}, 0, 0)$.

Second, we prove that $\poo(\bar{t})>\pzo(\bar{t})$ implies $\lim_{N\to
  \infty}\beta_{0|1}(\bar{t}, 0, 0)=0$. As a consequence, we have that
$\lim_{N\to \infty}\beta_{1|1}(\bar{t}, 0, 0)=1$ or equivalently that
$\lim_{N\to \infty}\, \pwof(N) = 0$.

To show equation~\eqref{eq:P1}, we consider the event \emph{the group is
  not giving the decision at time $\bar{t}$}. We aim to show that the
probability of this event goes to zero as $N\to\infty$. Indeed we have that
$$
\P\left[T\neq \bar{t}\right]=\P\left[T>\bar{t}\right]=1-
\left(\poo(\bar{t},N)+\pzo(\bar{t},N)\right), 
$$
and, hence, $\P\left[T>\bar{t}\right]=0$ implies
$\poo(\bar{t},N)+\pzo(\bar{t},N)=1$.  Observe that
$$
\P\left[T>\bar{t}\right]=\sum_{j=0}^{\lfloor \frac{N}{2}\rfloor}{N \choose
  2j}{2j \choose j} \poi(\bar{t})^j p_{0|i}(\bar{t})^j
\Big(1-\poi(\bar{t})-p_{0|i}(\bar{t})\Big)^{N-2j}.
$$
For simplicity of notation, let us denote $x:=\pzo(\bar{t})$ and
$y:=\pzo(\bar{t})$. We distinguish two cases, (i) $x\neq y$ and (ii) $x=y$.
\smallskip  

\emph{Case $x \neq y$.}  We show that in this case there exists
$\bar{\epsilon}>0$, depending only on $x$ and $y$, such that
\begin{equation}
  \label{eq:ine}
  {2j \choose j}x^j y^j < \left(x+y-\bar{\epsilon}\right)^{2j}, \quad
  \text{for all } j\geq 1. 
\end{equation}
First of all observe that, since ${2j \choose j}x^j y^j$ is just one term
of the Newton binomial expansion of $ \left(x+y\right)^{2j}$, we know that
$ {2j \choose j}x^j y^j < \left(x+y\right)^{2j}$ for all $j\in \N$.  Define
$\epsilon(j):=x+y-{2j \choose j}^{1/2j}\sqrt{xy}$ and observe that proving
equation~\eqref{eq:ine} is equivalent to proving $ \lim_{j\to \infty}
\epsilon(j)>0 $. Indeed if $\lim_{j\to \infty} \epsilon(j)>0$, then
$\inf_{j\in \N} \epsilon(j)>0$ and thereby we can define
$\bar{\epsilon}:=\inf_{j\in \N} \epsilon(j)$. To prove the inequality
$\lim_{j\to \infty} \epsilon(j)>0$, let us compute $\lim_{j\to \infty}{2j
  \choose j}^{1/(2j)}$. By applying Stirling's formula we can write
$$
  \lim_{j\to \infty}{2j \choose j}^{1/(2j)}=\lim_{j\to \infty}
  \left(\frac{\sqrt{2\pi2j}\left(\frac{2j}{e}\right)^{2j}}{2\pi j
      \left(\frac{j}{e}\right)^{2j}}\right)^{1/(2j)}=\left(\sqrt{\frac{1}{\pi
        j^2}}2^{2j}\right)^{1/(2j)}=2 
$$
and, in turn, $\lim_{j \to \infty} \epsilon(j)=x+y-2\sqrt{xy}$.  Clearly,
if $x\neq y$, then $x+y-2\sqrt{xy}>0$. Defining $\bar{\epsilon}:=\inf_{j\in
  \N} \epsilon(j)$, we can write
\begin{align*}
  \lim_{N\to \infty}  \sum_{j=0}^{\lfloor \frac{N}{2}\rfloor}{N \choose 2j}{2j \choose j}x^j y^j \left(1-x-y\right)^{N-2j} & \leq   \lim_{N\to \infty}  \sum_{j=0}^{\lfloor \frac{N}{2}\rfloor} {N \choose 2j}  \left(x+y-\bar{\epsilon}\right)^{2j} \left(1-x-y\right)^{N-2j} \\
  &\leq   \lim_{N\to \infty}  \sum_{j=0}^N {N \choose j}  \left(x+y-\bar{\epsilon}\right)^{j} \left(1-x-y\right)^{N-j} \\
  &= \lim_{N\to \infty} \left(1-\bar{\epsilon}\right)^{N}=0,
  \end{align*}
  which implies also $\lim_{N\to \infty}   \P\left[T>\bar{t}\right]=0$. 
\smallskip
  
\emph{Case $x = y$.}  To study this case, let $y=x+\xi$ and let $\xi \to
0$. In this case, the probability of the decision time exceeding $\bar{t}$
becomes
$$
f(x,N,\xi)= \P\left[T>\bar{t}\right]=\sum_{j=0}^{\lfloor
  \frac{N}{2}\rfloor}{N \choose 2j}{2j \choose j}x^{j}(x+\xi)^{j}
\left(1-2x-\xi\right)^{N-2j}.
$$
Consider $\lim_{\xi \to 0} f(x,N,\xi)$. We have that
\begin{align*}
  \lim_{\xi \to 0} f(x,N,\xi)= \sum_{j=0}^{\lfloor \frac{N}{2}\rfloor}{N \choose 2j}{2j \choose j}x^{2j} \left(1-2x\right)^{N-2j}< \sum_{j=0}^{\lfloor \frac{N}{2}\rfloor}{N \choose 2j} 2^{2j}x^{2j} \left(1-2x\right)^{N-2j} <1,
\end{align*}
where the first inequality follows from ${2j \choose j}<\sum_{j=0}^{2j} {2j
  \choose j} = 2^{2j}$, and the second inequality follows from
$\sum_{j=0}^{\lfloor \frac{N}{2}\rfloor}{N \choose 2j}
(2x)^{2j}<\sum_{j=0}^{N}{N \choose 2j} (2x)^{2j}=1$.  So $\lim_{\xi \to 0}
f(x,N,\xi)$ exists, and since we know that also $\lim_{N \to \infty}
f(x,N,\xi)$ exists, the limits are exchangeable in $\lim_{N \to \infty}
\lim_{\xi \to 0} f(x,N,\xi)$ and
\begin{equation*}
  \lim_{N \to \infty}  \lim_{\xi \to 0} f(x,N,\epsilon) = \lim_{\xi \to 0}
  \lim_{N \to \infty}  f(x,N,\xi) = 0. 
\end{equation*}
This concludes the proof of equation~\eqref{eq:P1}.
  
Assume now that $\poo(\bar{t})>\pzo(\bar{t})$. We distinguish between the
case where $\poo(\bar{t})>\frac{1}{2}$ and the case where
$\pzo(\bar{t})<\poo(\bar{t})\leq \frac{1}{2}$.
 
If $\poo(\bar{t}) > \frac{1}{2}$, then Lemma~\ref{Lemm-lim-comb-sum}
implies
\begin{equation*}
  \lim_{N \to \infty} \sum_{j=\lceil N/2 \rceil}^N {N \choose j}
  p^j_{1|1}(\bar{t}) \left(1-\poo(\bar{t}) \right)^{N-j} =1, 
\end{equation*}
and, since $\lim_{N\to \infty}\beta_{1|1}(\bar{t}, 0, 0)>\lim_{N \to
  \infty} \sum_{j=\lceil N/2 \rceil}^N {N \choose j} p^j_{1|1}(\bar{t})
\left(1-\poo(\bar{t}) \right)^{N-j}$, we have also that $\lim_{N\to
  \infty}\beta_{1|1}(\bar{t}, 0, 0)=1$.

The case $\pzo(\bar{t}) < \poo(\bar{t}) < \frac{1}{2}$ is more involved. We
will see that in this case $\lim_{N\to \infty}\beta_{0|1}(\bar{t}, 0,
0)=0$.  We start by observing that, from Lemma~\ref{Lemm-lim-comb-sum},
$$\lim_{N \to \infty} \sum_{j=\lceil \frac{N}{2} \rceil}^N {N \choose j}
p^j_{1|1}(\bar{t})\left([ 1-\poo(\bar{t}) \right)^{N-j}=0,$$ 
and in turn
\begin{equation*}
  \lim_{N \to \infty} \beta_{1|1}(\bar{t}, 0, 0) = \lim_{N \to \infty}
  \sum_{j=1}^{\lfloor \frac{N}{2} \rfloor} {N \choose j}p^j_{1|1}(\bar{t}) 
  \times \biggl(\sum_{i=0}^{j=1} {N-j \choose i} p^i_{0|1}(\bar{t}) \biggl[
  1-\poo(\bar{t})-\pzo(\bar{t}) \biggr]^{N-j-i}  \biggr). 
\end{equation*}  
The above expression can be written as follows
\begin{align*}
  \lim_{N \to \infty} \beta_{1|1}(\bar{t}, 0, 0)&=\lim_{N \to \infty}
  \sum_{h=1}^{N-2}\biggl( \sum_{j=\lfloor \frac{h}{2} \rfloor+1}^{h} {N
    \choose j} {N-j \choose h-j} p^{h-j}_{0|1}(\bar{t})
  p^{j}_{1|1}(\bar{t}) \biggr) \biggl( 1-\left( \pzo(\bar{t}) \poo(\bar{t})
  \right) \biggr)^{N-h}\\
  &=\lim_{N \to \infty} \sum_{h=1}^{N-2} {N \choose h}
  \sum_{j=\floor{\frac{h}{2}}+1}^h {h \choose j} p^{h-j}_{1|1}(\bar{t})
  p^j_{0|1}(\bar{t}) \biggl( 1-\poo(\bar{t})-\pzo(\bar{t}) \biggr)^{N-h}
\end{align*}
where, for obtaining the second equality we used the fact ${N \choose
  j}{N-j \choose h-j}= {N \choose h} {h \choose j}$. Similarly,
\begin{align*} 
  \lim_{N \to \infty} \beta_{0|1}(\bar{t}, 0, 0)&=\lim_{N \to \infty}
  \sum_{h=1}^{N-2} {N \choose h} \sum_{j=\floor{\frac{h}{2}}+1}^h {h
    \choose j} p^{h-j}_{0|1}(\bar{t}) p^j_{1|1}(\bar{t}) \biggl(
  1-\poo(\bar{t})-\pzo(\bar{t}) \biggr)^{N-h}.
\end{align*}
We prove now that $\lim_{N \to \infty} \beta_{0|1}(\bar{t}, 0, 0)=0$. To do
so we will show that there exists $\bar{\epsilon}$ depending only on
$\pzo(\bar{t})$ and $\poo(\bar{t})$ such that
\begin{equation*}
  \sum_{j=\floor{\frac{h}{2}}+1}^h {h \choose j} p^{h-j}_{0|1}(\bar{t})
  p^j_{1|1}(\bar{t}) < \biggl( \pzo(\bar{t}) + \poo(\bar{t}) - \bar{\epsilon}
  \biggr)^h.
\end{equation*}
To do so, let
\begin{equation*}
  \epsilon(h) = \pzo(\bar{t}) + \poo(\bar{t}) -
  \sqrt[h]{\sum_{j=\floor{\frac{h}{2}}+1}^{h} {h \choose j}
    p^{h-j}_{0|1}(\bar{t}) p^j_{1|1}(\bar{t})} .  
\end{equation*}
Because $h$ is bounded, one can see that $\epsilon(h)>0$ as the sum inside
the root is always smaller than $(\pzo(\bar{t})+\poo(\bar{t}))^h$. Also
\begin{align*}
\lim_{h \to \infty} \epsilon(h) &=  \biggl(  \pzo(\bar{t}) + \poo(\bar{t}) \biggr) \left[ 1- \frac{\sqrt[h]{\sum_{j=\floor{\frac{h}{2}}+1}^{h} {h \choose j} p^{h-j}_{0|1}(\bar{t}) p^j_{1|1}(\bar{t})}}{\pzo(\bar{t})+\poo(\bar{t})} \right]\\
&=  \biggl(  \pzo(\bar{t}) + \poo(\bar{t}) \biggr) \left[1- \sqrt[h]{\frac{\sum_{j=\floor{\frac{h}{2}}+1}^{h} {h \choose j} p^{h-j}_{0|1}(\bar{t}) p^j_{1|1}(\bar{t})}{\left(\pzo(\bar{t})+\poo(\bar{t})\right)^h}} \right] = \pzo(\bar{t})+\poo(\bar{t}),
\end{align*}
as by Lemma~\ref{Lemm-lim-comb-sum}, 
$$
\lim_{h \to \infty} \frac{\sum_{j=\floor{\frac{h}{2}}}^{h} {h \choose j} p^{h-j}_{0|1}(\bar{t}) p^j_{1|1}(\bar{t})}{\left(\pzo(\bar{t})+\poo(\bar{t})\right)^h}=0.
$$
Since by assumption, $\pzo(\bar{t})+\poo(\bar{t})>0$, we have that $\inf_{h
  \in \N} \epsilon(h)>0$. By letting $\bar{\epsilon}:=\inf_{h \in \N}
\epsilon(h)$, we conclude that
\begin{align*}
\lim_{N \to \infty} \beta_{0|1}(\bar{t},0,0) &\leq \sum_{h=1}^{N-2} {N \choose h} \biggl( \poo(\bar{t})+\pzo(\bar{t}) - \bar{\epsilon} \biggr) \biggl( 1-\poo(\bar{t})-\pzo(\bar{t}) \biggr)^{N-h}\\
&\leq \sum_{h=0}^{N}  {N \choose h} \biggl( \poo(\bar{t})+\pzo(\bar{t}) - \bar{\epsilon} \biggr) \biggl( 1-\poo(\bar{t})-\pzo(\bar{t}) \biggr)^{N-h} = (1-\bar{\epsilon})^N = 0.
\end{align*}
This concludes the proof.
\end{proof}

\begin{remark}
  The earliest possible decision time $\bar{t}$ defined in~\eqref{eq:bart}
  is the best performance that the fastest rule can achieve in terms of
  number of iterations required to provide the final decision.  \oprocend
\end{remark}

\subsection{The majority rule for varying values of $N$}\label{sec:FixedQ-majority}
We consider now the \emph{majority} rule, i.e., the counting rule with
threshold $q=\floor{N/2}+1$. We start with the following result about the
accuracy.  Recall that $\pwo$ is the probability of wrong decision by a
single SDM and that $\pwom(N)$ is the probability of wrong decision by a
group of $N$ SDMs implementing the majority rule (assuming $H_1$ is the
correct hypothesis).

\begin{proposition}[Accuracy under the majority rule]
  \label{prop:Acc_Maj}%
  Consider the $q$ out of $N$ SDA algorithm under Assumption
  \ref{ass:1}. Assume $q=\lfloor N/2 \rfloor +1$, i.e., the \emph{majority}
  rule is adopted. Without loss of generality, assume $H_1$ is the correct
  hypothesis. Then the probability of error satisfies
  \begin{equation}\label{eq:pmMaj}
    \pwom(N) = \sum_{j=\floor{N/2}+1}^{N} {N
      \choose j} \pwo^j \left(1-\pwo\right)^{N-j}. 
  \end{equation}
  According to~\eqref{eq:pmMaj}, the following characterization follows:
  \begin{enumerate}
  \item if $0\leq\pwo<1/2$, then $\pwom(N)$ is a monotonic
    decreasing function of $N$ that approaches $0$ asymptotically, that is,
    \begin{equation*} 
      \pwom(N)>\pwom(N+2)\quad
      \text{and}\quad \lim_{N\to \infty}\,
      \pwom(N)=0; 
    \end{equation*}
  \item if $1/2 < \pwo \leq 1$, then $\pwom(N)$ is a
    monotonic increasing function of $N$ that approaches $1$
    asymptotically, that is,
    \begin{equation*}\label{eq:pmMaj2}
      \pwom(N)<\pwom(N+2)\quad
      \text{and}\quad \lim_{N\to \infty}\,
      \pwom(N)=1;
    \end{equation*}
  \item if $\pwo=1/2$, then $\pwom(N)=1/2$;

  \item if $\pwo <1/4$, then
    \begin{equation}\label{eq:pmMaj4}
      \pwom(N)= {N \choose \lceil\frac{N}{2}\rceil}
      \,\pwo^{\lceil\frac{N}{2}\rceil} 
      +
      o\left(\pwo^{\lceil\frac{N}{2}\rceil}\right)
      =
      \sqrt{N/(2\pi)}\, 
      (4\pwo)^{\lceil\frac{N}{2}\rceil}
      +
      o\left( (4\pwo)^{\lceil\frac{N}{2}\rceil}\right)
      .
    \end{equation}
  \end{enumerate}
\end{proposition}
\begin{proof}
  We start by observing that
  $$
  \sum_{s=1}^t \pzo(s; N, q=\floor{N/2}+1)=\sum_{j= \floor{N/2}+1}^{N} {N
    \choose j} \pi_{0|1}(t)^j \left(1-\pi_{0|1}(t)\right)^{N-j}.
  $$
  Since $ \pwom(N)=\sum_{s=1}^{\infty} \pzo(s; N,
  q=\floor{N/2}+1)$, taking the limit for $N\to \infty$ in the above
  expression leads to
  $$
  \pwom(N)=\sum_{j=\lceil \frac{N}{2}\rceil}^{N} {N \choose j}
  \pwo^j \left(1-\pwo\right)^{N-j}.
  $$
  Facts (i), (ii), (iii) follow directly from Lemma~\ref{Lemm-lim-comb-sum}
  in Appendix~\ref{appendix:combinatorial-sums} applied to
  equation~\eqref{eq:pmMaj}.  Equation~\eqref{eq:pmMaj4} is a consequence
  of the Taylor expansion of~\eqref{eq:pmMaj}:
  \begin{align*}
    \sum_{j=\lceil \frac{N}{2} \rceil}^{N} {N \choose j}
    \pwo^j(1-\pwo)^{N-j}&=\sum_{j=\lceil \frac{N}{2} \rceil}^{N} {N \choose
      j} \pwo^j (1-(N-j)\pwo+ o(\pwo))\\
    &={N \choose \lceil\frac{N}{2}\rceil} \,\pwo^{\lceil\frac{N}{2}\rceil} +
    o\left(\pwo^{\lceil\frac{N}{2}\rceil}\right).
  \end{align*}
  Finally, Stirling's Formula implies $\lim_{N\to \infty}{N \choose
    \lceil\frac{N}{2}\rceil} = \sqrt{2N/\pi}\,2^N$ and, in turn, the final
  expansion follows from $2^N=4^{\ceil{N/2}}/2$.
\end{proof}


We discuss now the expected time required by the collective SDA algorithm
to provide a decision when the \emph{majority} rule is adopted. Our
analysis is based again on Assumption~\ref{ass:1} and on the assumption
that $H_1$ is the correct hypothesis. We distinguish four cases based on
different properties that the probabilities of wrong and correct decision
of the single SDM might have:
\begin{enumerate}
\item[(A1)] the probability of correct decision is greater than the
  probability of wrong decision, i.e., $\pco > \pwo$;
\item[(A2)] the probability of correct decision is equal to the probability
  of wrong decision, i.e., $\pco=\pwo=1/2$ and there exist $t_0$ and $t_1$
  such that $\pi_{0|1}(t_0)=1/2$ and $\pi_{1|1}(t_1)=1/2$;
\item[(A3)] the probability of correct decision is equal to the probability
  of wrong decision, i.e., $\pco=\pwo=1/2$ and there exists $t_1$ such that
  $\pi_{1|1}(t_1)=1/2$, while $\pi_{0|1}(t)<1/2$ for all $t\in \N$ and
  $\lim_{t\to \infty}\pi_{0|1}(t)=1/2$;
\item[(A4)] the probability of correct decision is equal to the probability
  of wrong decision, i.e., $\pco=\pwo=1/2$, and $\pi_{0|1}(t)<1/2$,
  $\pi_{1|1}(t)<1/2$ for all $t\in \N$ and $\lim_{t\to
    \infty}\pi_{0|1}=\lim_{t\to \infty} \pi_{1|1}(t)=1/2$.
\end{enumerate}

Note that, since Assumption~\ref{ass:1} implies $\pco+\pwo=1$, the
probability of correct decision in case (A1) satisfies $\pco>1/2$.  Hence,
in case (A1) and under Assumption~\ref{ass:1}, we define
$t_{<\frac{1}{2}}:=\max\setdef{t\in\N}{\pi_{1|1}(t)<1/2}$ and
$t_{>\frac{1}{2}}:=\min\setdef{t\in\N}{\pi_{1|1}(t)>1/2}$.

\begin{proposition}[Expected time under the \emph{majority} rule]
  \label{prop:ET_Maj}
  Consider the $q$ out of $N$ SDA algorithm under Assumption
  \ref{ass:1}. Assume $q=\lfloor N/2 \rfloor +1$, that is, adopt the
  \emph{majority} rule. Without loss of generality, assume $H_1$ is the
  correct hypothesis. Define the SDM properties (A1)-(A4) and the decision
  times $t_0$, $t_1$, $t_{<\frac{1}{2}}$ and $t_{> \frac{1}{2}}$ as
  above. Then the expected decision time satisfies
  \begin{equation*}
    \lim_{N \to \infty} \E\big[T| H_1,N, q=\ceil{N/2} \big]=
    \begin{cases}
      \displaystyle\frac{t_{<\frac{1}{2}}+t_{>\frac{1}{2}}+1}{2},%
      \quad &\text{if the SDM has the property (A1),}\\[.35em]
      \displaystyle \frac{t_1+t_0}{2}, \quad%
      \quad &\text{if the SDM has the property (A2),} \\[.35em]
      \displaystyle +\infty, \quad 
      \quad &\text{if the SDM has the property (A3) or (A4).}
    \end{cases}
  \end{equation*}
\end{proposition}
\begin{proof}
  We start by proving the equality for case (A1). Since, in this case we
  are assuming $\pco > \pwo$, the definitions of $t_{<\frac{1}{2}}$ and
  $t_{> \frac{1}{2}}$ implies that $\pi_{1|1}(t)=1/2$ for all
  $t_{<\frac{1}{2}} < t <t_{> \frac{1}{2}}$. Observe that
  \begin{equation*}
    \sum_{s=1}^t \poo(t;N,q=\floor{N/2}+1)=\sum_{h=\floor{\frac{N}{2}}}^{N}
    {N \choose h} \pi^h_{1|1}(t) \biggl( 1- \pi_{1|1}(t) \biggr)^{N-h}. 
  \end{equation*}
  Hence Lemma~\ref{Lemm-lim-comb-sum} implies
  $$
  \lim_{N \to \infty}   \sum_{s=1}^t \poo(t;N,q=\floor{N/2}+1)=
  \begin{cases}
    0, \quad & \text{if } \; t\leq t_{<\frac{1}{2}}, \\
    1, \quad & \text{if } \; t \geq t_{> \frac{1}{2}}, \\
    \frac{1}{2}, \quad & \text{if } t_{<\frac{1}{2}}< t < t_{>\frac{1}{2}}\;,
  \end{cases}
  $$
  and, in turn, that
  $$
  \lim_{N \to \infty}  \poo(t;N,q=\floor{N/2}+1)=
  \begin{cases}
    1/2, \quad & \text{if } \; t= t_{<\frac{1}{2}}+1\,\,\,\,
    \text{and}\,\,\,\,t=t_{> \frac{1}{2}} , \\
    \,\,0, \quad & \text{otherwise}.
  \end{cases}
  $$
  It follows
  \begin{align*}
    \lim_{N \to \infty} \E\left[T| H_1,N, q=\floor{N/2}+1 \right]&= \lim_{N
      \to \infty} t
    \left(\pzo(t;N,q=\floor{N/2}+1)+\poo(t;N,q=\floor{N/2}+1)\right)\\
    &=\frac{1}{2}\left(t_{<\frac{1}{2}}+1+t_{> \frac{1}{2}}
    \right).
  \end{align*}
  This concludes the proof of the equality for case (A1).
  
  We consider now the case (A2). Reasoning similarly to the previous case
  we have that
  $$
  \lim_{N\to \infty} \poo(t_1; N, q=\floor{N/2}+1)=1/2 \qquad \text{and}
  \qquad \lim_{N\to \infty} \pzo(t_0; N, q=\floor{N/2}+1)=1/2,
  $$
  from which it easily follows that $ \lim_{N \to \infty} \E\left[T| H_1,N,
    q=\floor{N/2}+1 \right]=\frac{1}{2}\left(t_0+t_1\right)$.
 
  For case (A3), it suffices to note the following implication of
  Lemma~\ref{Lemm-lim-comb-sum}:
  if, for a given $i\in\{0,1\}$, we have $\pi_{i|1}(t)< 1/2$ for all $t\in
  \N$, then $\lim_{N\to \infty}\pio(t; N, q=\floor{N/2}+1)=0$ for all $t\in
  \N$.  The analysis of the case (A4) is analogous to that of case (A3).
\end{proof}

\begin{remark}
  The cases where $\pwo>\pco$ and where there exists $t_0$ such that
  $\pi_{0|1}(t_0)=1/2$ while $\pi_{1|1}(t)< 1/2$ for all $t\in \N$
  and $\lim_{t \to \infty} \pi_{1|1}(t)=1/2$, can be analyzed similarly to
  the cases (A1) and (A3).  Moreover, the most recurrent situation in
  applications is the one where there exists a time instant $t$ such
  that $\pi_{1|1}(t)< 1/2$ and $\pi_{1|1}(t+1)>1/2$, which is
  equivalent to the above case (A1) with $t_{>
    \frac{1}{2}}=t_{< \frac{1}{2}} +1$. In this situation we
  trivially have $ \lim_{N \to \infty} \E\left[T| H_1,N, q=\lceil N/2
    \rceil \right]=t_{> \frac{1}{2}}$.  \oprocend
\end{remark}

\subsection{Fixed $N$ and varying $q$}\label{sec:FixedN}

We start with a simple result characterizing the expected decision time. 
\begin{proposition}\label{prop:ETq}
  Given a group of $N$ SDMs running the \emph{$q$ out of $N$ SDA}, for
  $j\in \{0,1\}$,
  \begin{equation*}
    \E[T|H_j, N, q=1]\leq \E[T|H_j, N, q=2] \leq \dots \leq \E[T|H_j, N, q=N].    
  \end{equation*}
\end{proposition}

The above proposition states that the expected number of iterations required to provide a decision constitutes a nondecreasing sequence for increasing value of $q$. Similar monotonicity results hold true also for $\pcj(N,q)$, $\pwj(N,q)$, $\pndj(N,q)$ even though restricted only to $\lfloor N/2 \rfloor+1 \leq q \leq N$. 

\begin{proposition}
  Given a group of $N$ SDMs running the \emph{$q$ out of $N$ SDA}, for
  $j\in \{0,1\}$,
  \begin{align*}
    \pcj(N,q=\floor{N/2}+1) \geq  \pcj(N,q=\floor{N/2} +2) &\geq \dots \geq \pcj(N,q=N),\\
    \pwj(N,q=\floor{N/2}+1) \geq \pwj(N,q=\floor{N/2}  +2) &\geq \dots \geq \pwj(N,q=N),\\
    \pndj(N,q=\floor{N/2}+1) \leq \pndj(N,q=\floor{N/2} +2 )&\leq \dots
    \leq \pndj(N,q=N).
  \end{align*}
\end{proposition}

We believe that similar monotonic results hold true also for $1 \leq q \leq
\floor{N/2}$. In particular, here is our conjecture: if $N$ is odd, the
single SDM has the \emph{almost-sure} decision and the single SDM is more
likely to provide the correct decision than the wrong decision, that is,
$\pcj+\pwj=1$ and $\pcj>\pwj$, then
\begin{align*}
  \pcj(N,q=1)\leq  \pcj(N,q=2) &\leq \dots \leq \pcj(N,q=\floor{N/2}+1),\\
  \pwj(N,q=1)\geq \pwj(N,q=2) &\geq \dots \geq \pwj(N,q=\floor{N/2}+1).
\end{align*}
These chains of inequalities are numerically verified in some examples in
Section~\ref{SSec-sim-nocomm}.

\section{Numerical analysis}\label{SSec-sim-nocomm}

The goal of this section is to numerically analyze the models and methods
described in previous sections.  In all the examples, we assume that the
sequential binary test run by each SDMs is the classical sequential
probability ratio test (SPRT) developed in 1943 by Abraham Wald.  To fix
some notation, we start by briefly reviewing the SPRT.  Let $X$ be a random
variable with distribution $f(x; \theta)$ and assume the goal is to test
the null hypothesis $H_0: \theta=\theta_0$ against the alternative
hypothesis $H_1: \theta=\theta_1$. For $i\in\until{N}$, the $i$-th SDM
takes the observations $x_i(1), x_i(2), x(3), \ldots,$ which are assumed to
be independent of each other and from the observations taken by all the
other SDMs.  The log-likelihood ratio associated to the observation
$x_i(t)$ is
\begin{equation}\label{eq:log-lik}
  \lambda_i(t)=\log \frac{f(x_i(t), \theta_1)}{f(x_i(t), \theta_0)}.
\end{equation}
Accordingly, let $\Lambda_i(t)=\sum_{h=1}^t\lambda_i(h)$ denote the sum of
the log-likelihoods up to time instant $t$. The $i$-th SDM continues to
sample as long as $\eta_0 < \Lambda_i(t) < \eta_1$, where $\eta_0$ and
$\eta_1$ are two pre-assigned thresholds; instead sampling is stopped the
first time this inequality is violated. If $\Lambda_i(t)<\eta_0$, then the
$i$-th SDM decides for $\theta=\theta_0$. If $\Lambda_i(t)>\eta_1$, then
the $i$-th SDM decides for $\theta=\theta_1$.

\newcommand{\pmisd}{\subscr{p}{misdetection}}
\newcommand{\pfals}{\subscr{p}{false alarm}}

To guarantee the \emph{homogeneity property} we assume that all the SDMs
have the same thresholds $\eta_0$ and $\eta_1$.  The threshold values are
related to the accuracy of the SPRT as described in the classic Wald's
method~\cite{AW:45}.  We shortly review this method next.  Assume that, for
the single SDM, we want to set the thresholds $\eta_0$ and $\eta_1$ in such
a way that the probabilities of misdetection (saying $H_0$ when $H_1$ is
correct, i.e., $\P[\text{say } H_0 | H_1]$) and of false alarm (saying
$H_1$ when $H_0$ is correct, i.e., $\P[\text{say } H_1|H_0]$) are equal to
some pre-assigned values $\pmisd$ and $\pfals$.  Wald proved that the
inequalities $\P[\text{say } H_0\,|\,H_1] \leq \pmisd$ and $\P[\text{say }
H_1\,|\,H_0] \leq \pfals$ are achieved when $\eta_0$ and $\eta_1$ satisfy
$\eta_0 \leq \log \frac{\pmisd}{1-\pfals}$ and $\eta_1 \geq \log
\frac{1-\pmisd}{\pfals}$.  As customary, we adopt the equality sign in
these inequalities for the design of $\eta_0$ and $\eta_1$. Specifically,
in all our examples we assume that $\pmisd=\pfals=0.1$ and, in turn, that
$\eta_1=-\eta_0=\log 9$.


We provide numerical results for observations described by both discrete
and continuous random variables.  In case $X$ is a discrete random
variable, we assume that $f(x; \theta)$ is a binomial distribution
\begin{equation}\label{eq:bin}
  f(x;\theta)=
  \begin{cases}
    {n \choose x} \theta^x (1-\theta)^{n-x}, \quad & \text{if }
    x\in\left\{0,1,\ldots,n\right\},\\
    0, & \text{otherwise,}
  \end{cases}
\end{equation}
where $n$ is a positive integer number. In case $X$ is a continuous random
variable, we assume that $f(x; \theta)$ is a Gaussian distribution with
mean $\theta$ and variance $\sigma^2$
\begin{equation}\label{eq:gaus}
  f(x; \theta)=\frac{1}{\sqrt{2\pi \sigma^2}}
  e^{-{(x-\theta)^2}/{2\sigma^2}}.
\end{equation}

The key ingredient required for the applicability of
Propositions~\ref{Prop-Numq} and~\ref{prop:AlphaBeta} is the knowledge of
the probabilities $ \left\{ \pzz (t), \poz(t)\right\}_{t\in \N}$ and
$\left\{\pzo(t), \poo(t)\right\}_{t\in \N}$.  Given thresholds $\eta_0$ and
$\eta_1$, there probabilities can be computed according to the method
described in the Appendix~\ref{subsec:AccTimeDiscrete} (respectively
Appendix~\ref{subsec:AccTimeCont}) for $X$ discrete (respectively $X$
continuous) random variable.

We provide three sets of numerical results. Specifically, in
Example~\ref{ex:1} we emphasize the tradeoff between accuracy and expected
decision time as a function of the number of SDMs. In Example~\ref{ex:2} we
concentrate on the monotonic behaviors that the \emph{$q$ out of $N$} SDA
algorithm exhibits both when $N$ is fixed and $q$ varies and when $q$ is
fixed and $N$ varies. In Example~\ref{ex:3} we compare the \emph{fastest}
and the \emph{majority} rule.  Finally, Section~\ref{sec:DMCPR} discusses
drawing connections between the observations in Example~\ref{ex:3} and the
cognitive psychology presentation introduced in Section~\ref{ssection-cog}.

\begin{example}[Tradeoff between accuracy and expected decision
  time]\label{ex:1}
  This example emphasizes the tradeoff between accuracy and expected
  decision time as a function of the number of SDMs. We do that for the
  \emph{fastest} and the \emph{majority} rules.  We obtain our numerical
  results for odd sizes of group of SDMs ranging from $1$ to $61$. In all
  our numerical examples, we compute the values of the thresholds $\eta_0$
  and $\eta_1$ according to Wald's method by posing $\pmisd=\pfals=0.1$
  and, therefore, $\eta_1=\log 9$ and $\eta_0=-\log 9$.

  For a binomial distribution $f(x; \theta)$ as in~\eqref{eq:bin}, we
  provide our numerical results under the following conditions: we set
  $n=5$; we run our computations for three different pairs $(\theta_0,
  \theta_1)$; precisely we assume that $\theta_0=0.5-\epsilon$ and
  $\theta_1=0.5+\epsilon$ where $\epsilon \in
  \left\{0.02,0.05,0.08\right\}$; and $H_1: \theta=\theta_1$ is always the
  correct hypothesis.
  For any pair $(\theta_0, \theta_1)$ we perform the following three
  actions in order
  \begin{enumerate}
  \item we compute the probabilities $\left\{\pzo(t), \poo(t)\right\}_{t\in
      \N}$ according to the method described in
    Appendix~\ref{subsec:AccTimeDiscrete};
  \item we compute the probabilities $\left\{\pzo(t; N,q), \poo(t;
      N,q)\right\}_{t\in \N}$ for $q=1$ and $q=\floor{N/2}+1$ according to
    the formulas reported in Proposition~\ref{Prop-Numq};
  \item we compute probability of wrong decision and expected time for the
    group of SDMs exploiting the formulas
    \begin{equation*}
      \pwo(N,q)=\sum_{t=1}^{\infty}\pzo(t; N,q) %
      \quad\text{and}\quad%
      \E[T|H_1,N,q]=\sum_{t=1}^{\infty}(\pzo(t;N,q)+\poo(t;N,q))t.
    \end{equation*}
  \end{enumerate}

  According to Remark~\ref{rem:ASD}, since we consider only odd numbers $N$
  of SDMs, since $q \leq \lceil N/2 \rceil$ and since each SDM running the
  SPRT has the \emph{almost-sure decisions} property, then
  $\pwo(N,q)+\pco(N,q)=1$. In other words, the probability of no-decision
  is equal to $0$ and, hence, the accuracy of the SDA algorithms is
  characterized only by the probability of wrong decision and the
  probability of correct decision. In our analysis we select to compute the
  probability of wrong decision.

  For a Gaussian distribution $f(x; \theta, \sigma)$, we obtain our
  numerical results under the following conditions: the two hypothesis are
  $H_0: \theta=0$ and $H_1: \theta=1$; we run our computations for three
  different values of $\sigma$; precisely $\sigma \in
  \left\{0.5,1,2\right\}$; and $H_1: \theta=1$ is always the correct
  hypothesis.
  To obtain $\pwo(N,q)$ and $\E[T|H_1,N,q]$ for a given value of $\sigma$,
  we proceed similarly to the previous case with the only difference that
  $\left\{\pzo(t), \poo(t)\right\}_{t\in \N}$ are computed according to the
  procedure described in Appendix~\ref{subsec:AccTimeCont}.

  The results obtained for the \emph{fastest} rule are depicted in
  Figure~\ref{fig:fastest}, while the results obtained for the
  \emph{majority} rule are reported in Figure~\ref{fig:majority}.

  \begin{figure}[h!]
    \begin{center}
      {\includegraphics[width=.495\textwidth]{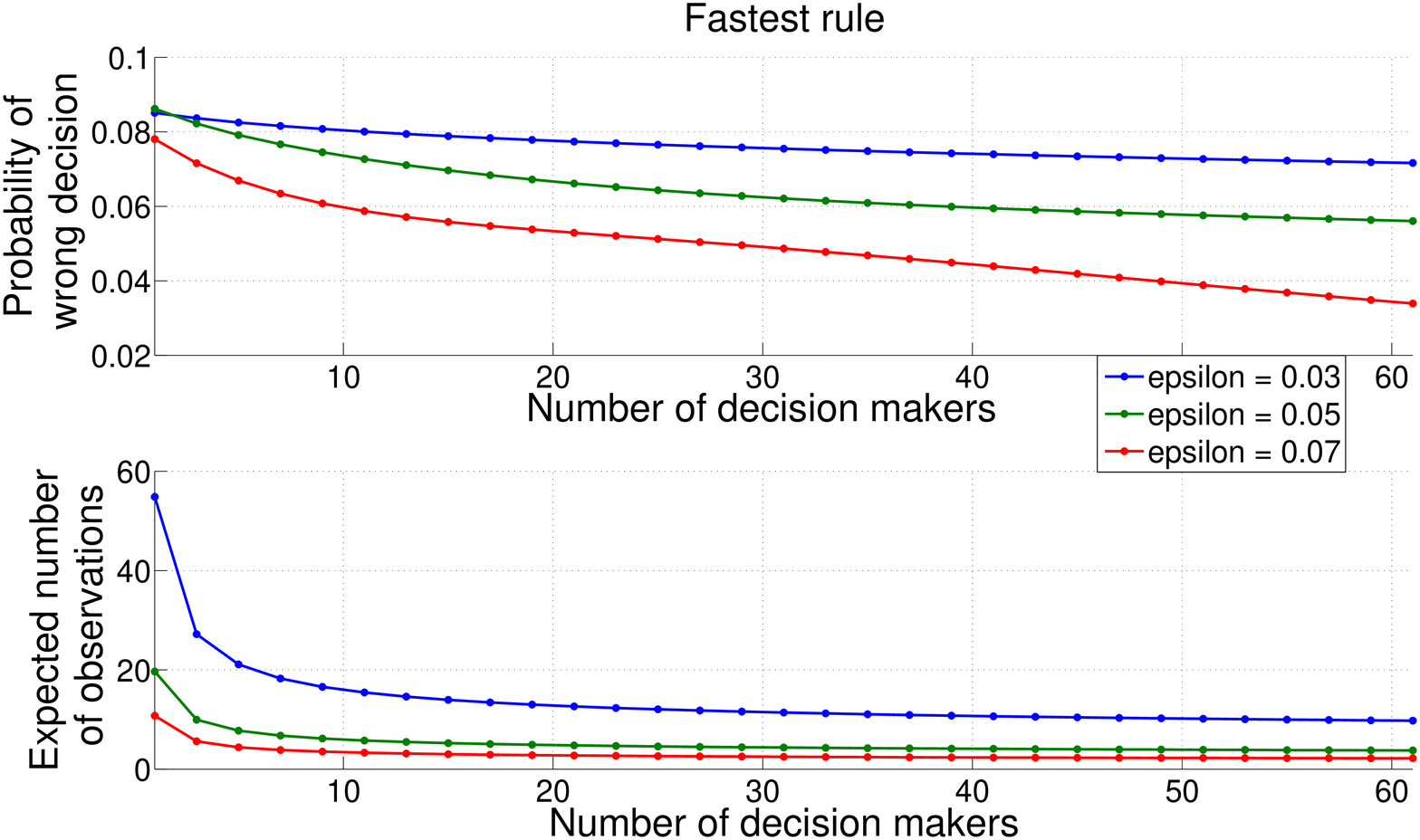}}
      {\includegraphics[width=.495\textwidth]{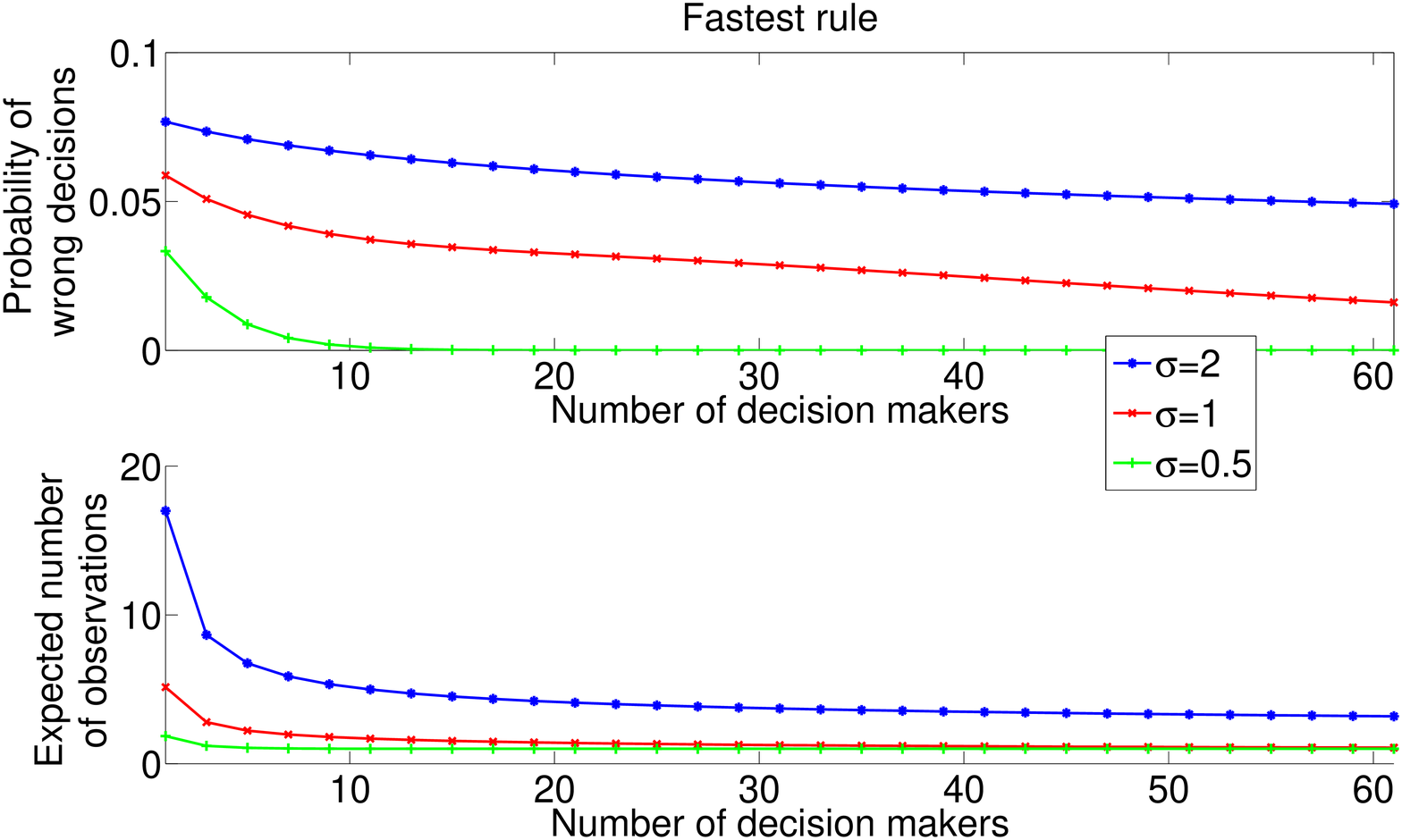}}
      \label{fig:fastest}
      \caption{Behavior of the probability of wrong decision and of the
        expected number of iterations required to provide a decision as the
        number of SDMs increases when the \emph{fastest} rule is
        adopted. In Figure (a) we consider the binomial distribution, in
        Figure (b) the Gaussian distribution.}
    \end{center}
  \end{figure}
  
  \begin{figure}[h!]
    \begin{center}
      {\includegraphics[width=.495\textwidth]{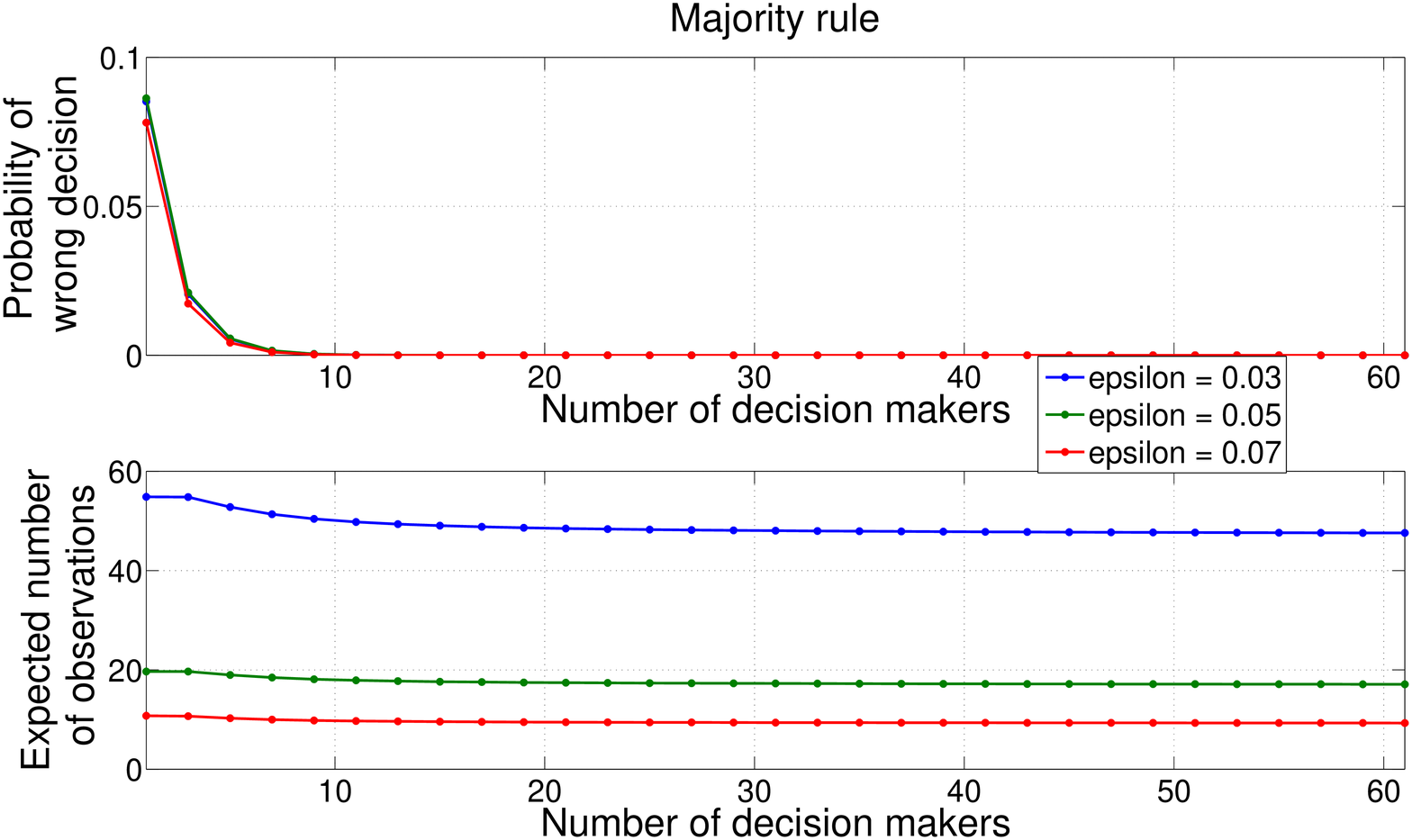}}%
      {\includegraphics[width=.495\textwidth]{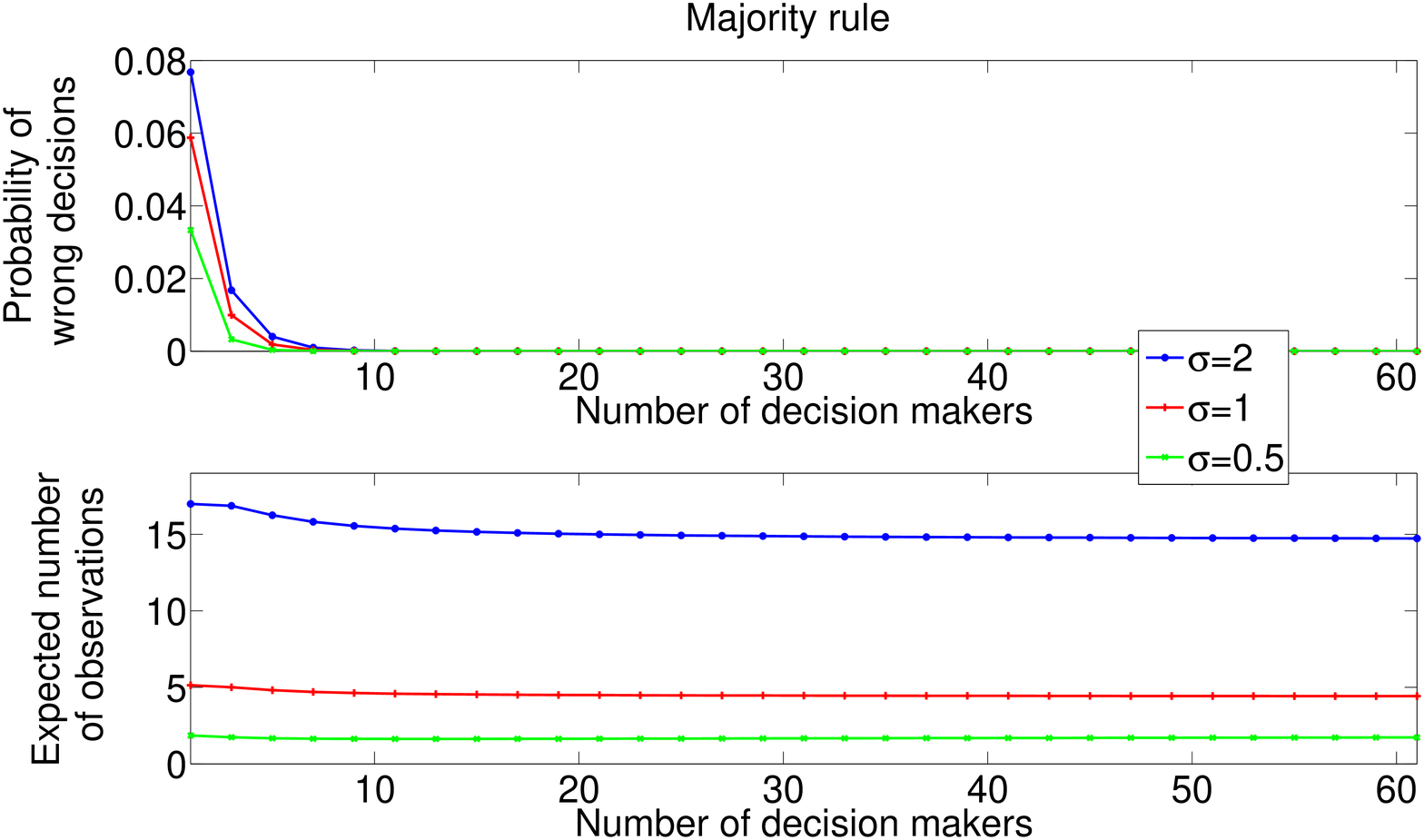}}
      \caption{Behavior of the probability of wrong decision and of the
        expected number of iterations required to provide a decision as the
        number of SDMs increases when the \emph{majority} rule is
        adopted. In Figure (a) we consider the binomial distribution, in
        Figure (b) the Gaussian distribution.}
      \label{fig:majority}
    \end{center}
  \end{figure}
  
  Some remarks are now in order. We start with the \emph{fastest} rule. A
  better understanding of the plots in Figure~\ref{fig:fastest} can be
  gained by specifying the values of the earliest possible decision time
  $\bar{t}$ defined in~\eqref{eq:bart} and of the probabilities
  $\poo(\bar{t})$ and $\pzo(\bar{t})$. In our numerical analysis, for each
  pair $(\theta_0, \theta_1)$ considered and for both discrete and
  continuous measurements $X$, we had $\bar{t}=1$ and
  $\poo(\bar{t})>\pzo(\bar{t})$. As expected from
  Proposition~\ref{prop:A-ET-Fastest}, we can see that the \emph{fastest}
  rule significantly reduces the expected number of iterations required to
  provide a decision. Indeed, as $N$ increases, the expected decision time
  $\E[T|H_1,N,q=1]$ tends to $1$. Moreover, notice that $\pwof(N)$
  approaches $0$; this is in accordance with equation~\eqref{eq:AccFast}.

  As far as the \emph{majority} rule is concerned, the results established
  in Proposition~\ref{prop:Acc_Maj} and in Proposition~\ref{prop:ET_Maj}
  are confirmed by the plots in Figure~\ref{fig:majority}. Indeed, since
  for all the pairs $(\theta_0, \theta_1)$ we have considered, we had
  $\pwo<1/2$, we can see that, as expected from
  Proposition~\ref{prop:Acc_Maj}, the probability of wrong decision goes to
  $0$ exponentially fast and monotonically as a function of the size of the
  group of the SDMs. Regarding the expected time, in all the cases, the
  expected decision time $\E[T|H_1,N,q=\floor{N/2}+1]$ quickly reaches a
  constant value. We numerically verified that these constant values
  corresponded to the values predicted by the results reported in
  Proposition~\ref{prop:ET_Maj}.
\end{example}

\begin{example}[Monotonic behavior]\label{ex:2}
  In this example, we analyze the performance of the general \emph{$q$ out
    of $N$} aggregation rule, as the number of SDMs $N$ is varied, and as
  the aggregation rule itself is varied.  We obtained our numerical results
  for odd values of $N$ ranging from $1$ to $35$ and for values of $q$
  comprised between $1$ and $\floor{N/2}+1$. Again we set the thresholds
  $\eta_0$ and $\eta_1$ equal to $\log(-9)$ and $\log 9$, respectively. In
  this example we consider only the Gaussian distribution with
  $\sigma=1$. The results obtained are depicted in
  Figure~\ref{fig:Pc_qoutN+ET_qoutN}, where the following monotonic
  behaviors appear evident:
  \begin{enumerate}
  \item for fixed $N$ and increasing $q$, both the probability of correct
    decision and the decision time increases;
  \item for fixed $q$ and increasing $N$, the probability of correct
    decision increases while the decision time decreases.
  \end{enumerate}
  The fact that the decision time increases for fixed $N$ and increasing
  $q$ has been established in Proposition~\ref{prop:ETq}.  The fact that
  the probability of correct decision increases for fixed $N$ and
  increasing $q$ validates the conjecture formulated at the end of
  Section~\ref{sec:FixedN}.
  \begin{figure}[h!]
    \begin{center}
      \includegraphics[width=.495\textwidth]{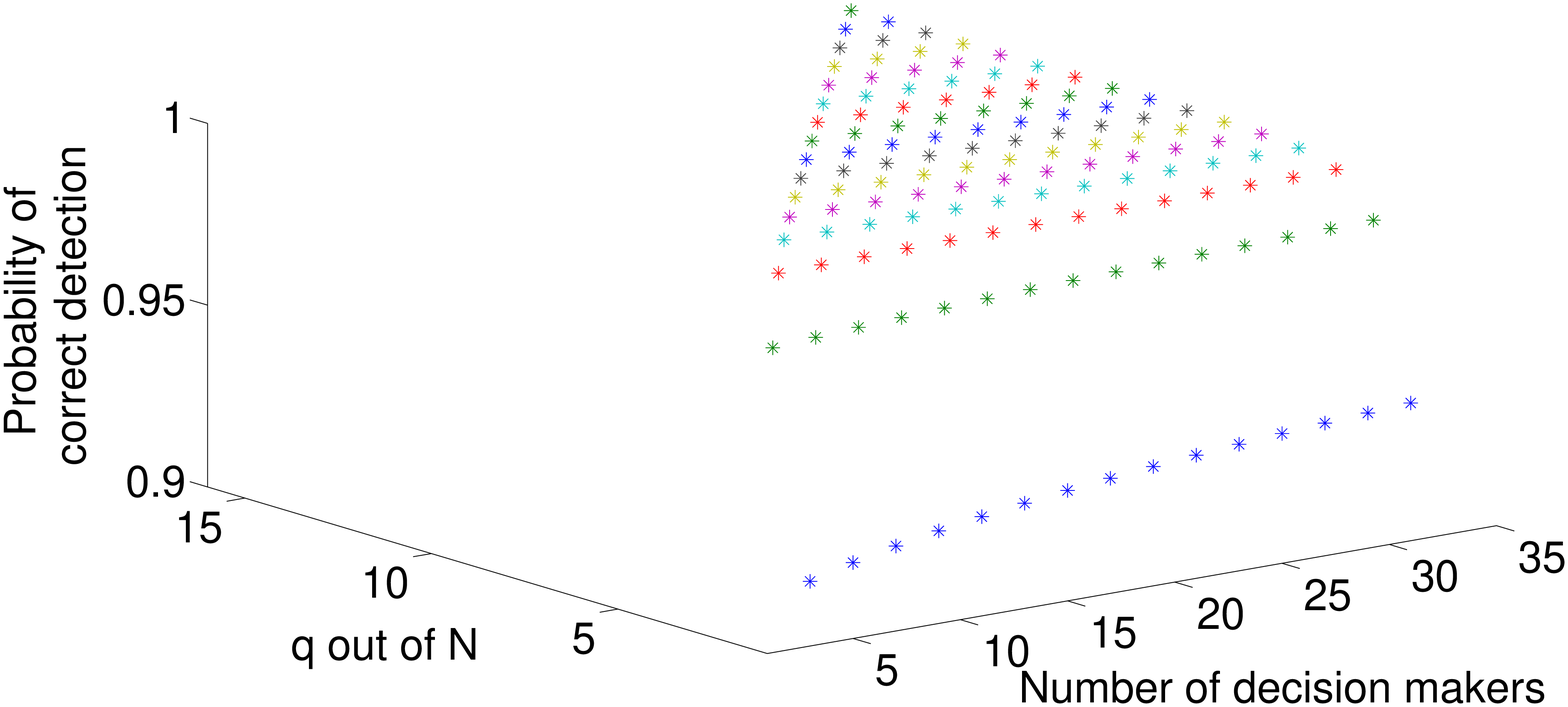}
      \includegraphics[width=.495\textwidth]{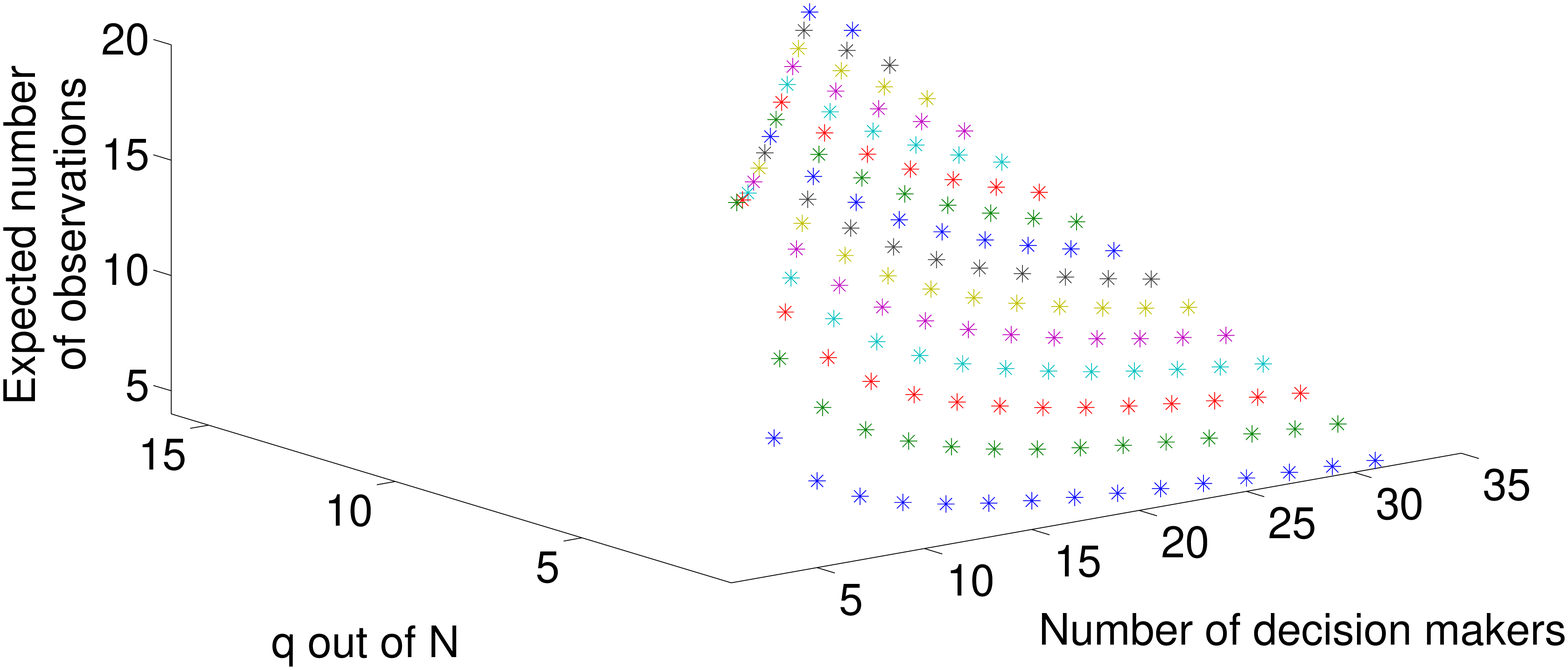}
      \caption{Probability of correct detection (left figure) and expected
        decision time (right figure) for the \emph{$q$ out of $N$} rule,
        plotted as a function of network size $N$ and accuracy threshold
        $q$.} \label{fig:Pc_qoutN+ET_qoutN}
    \end{center}
  \end{figure}
\end{example}

\begin{example}[Fastest versus majority, at fixed group
  accuracy]\label{ex:3}
  As we noted earlier, Figures~\ref{fig:fastest}-\ref{fig:majority} show
  that the \emph{majority} rule increases remarkably the accuracy of the
  group, while the \emph{fastest} rule decreases remarkably the expected
  number of iteration for the SDA to reach a decision. It is therefore
  reasonable to pose the following question: if the local accuracies of the
  SDMs were set so that the accuracy of the group is the same for both
  the \emph{fastest} and the \emph{majority} fusion rule, which of the two
  rules requires a smaller number of observations to give a decision. That
  is, at equal accuracy, which of the two rules is optimal as far as
  decision time is concerned.

  In order to answer this question, we use a bisection on the local SDM
  accuracies. We apply the numerical methods presented in
  Proposition~\ref{Prop-Numq} to find the proper local thresholds that set
  the accuracy of the group to the desired value $\pwo$. Different local
  accuracies are obtained for different fusion rules and this evaluation
  needs to be repeated for each group size $N$.


  In these simulations, we assume the random variable $X$ is Gaussian with
  variance $\sigma=2$. The two hypotheses are $H_0:\theta=0$ and
  $H_1:\theta=1$. The numerical results are shown in
  Figure~\ref{fig-netaccDT} and discussed below.

  As is clear by the plots, the strategy that gives the fastest decision
  with the same accuracy varies with group size and desired accuracy. The
  left plot in Figure~\ref{fig-netaccDT} illustrates that, for very high
  desired group accuracy, the \emph{majority} rule is always optimal. As
  the accuracy requirement is relaxed, the \emph{fastest} rule becomes
  optimal for small groups. Moreover, the group size at which the switch
  between optimal rules happens, varies for different accuracies. For
  example, the middle and right plot in Figure~\ref{fig-netaccDT}
  illustrate that while the switch happens at $N=5$ for a group accuracy
  $\pwom=\pwof=0.05$ and at $N=9$ for $\pwom=\pwof=0.1$.
  
  \begin{figure}[h!]
    \begin{center}
      {\includegraphics[width=.3\textwidth]{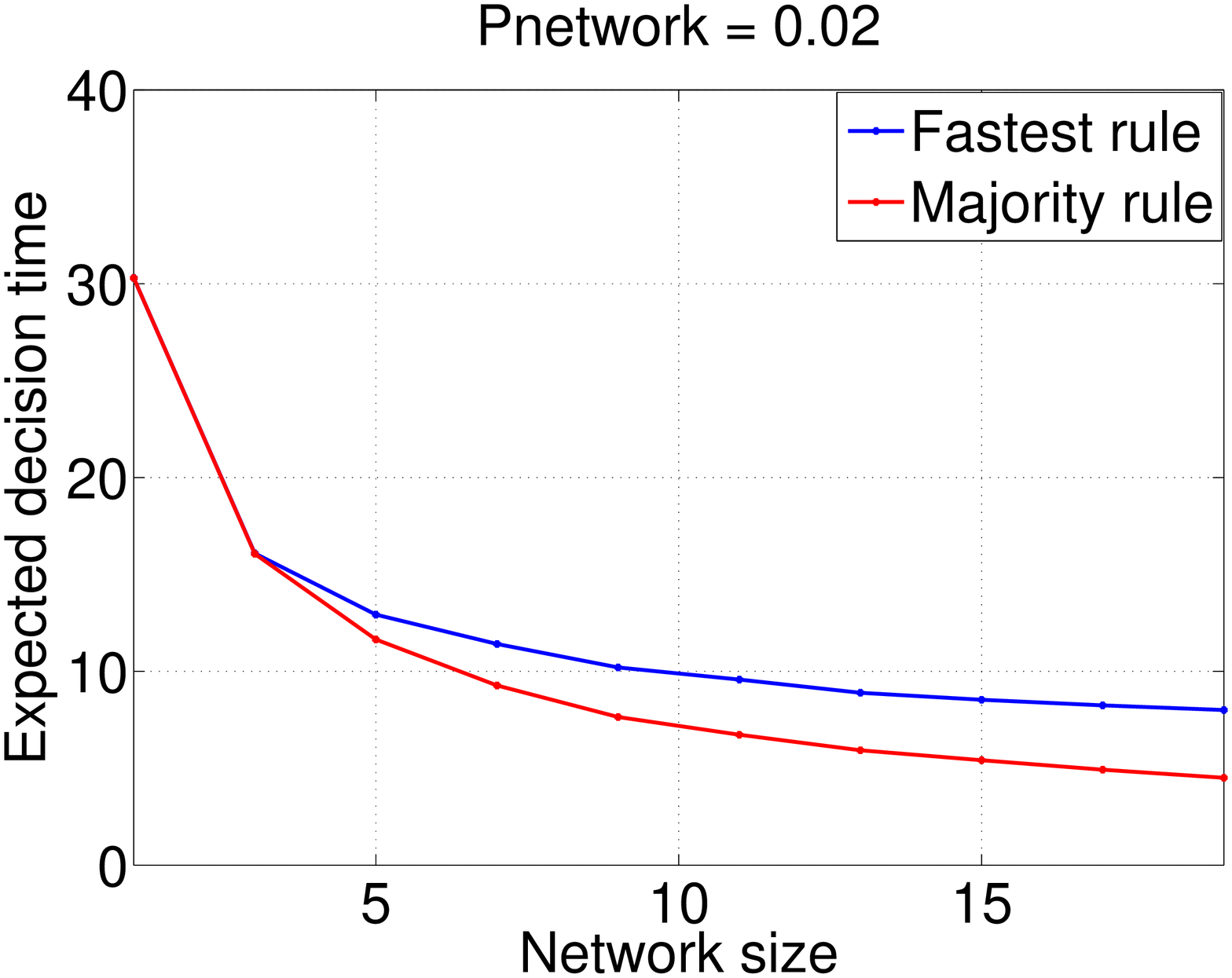}}
      {\includegraphics[width=.3\textwidth]{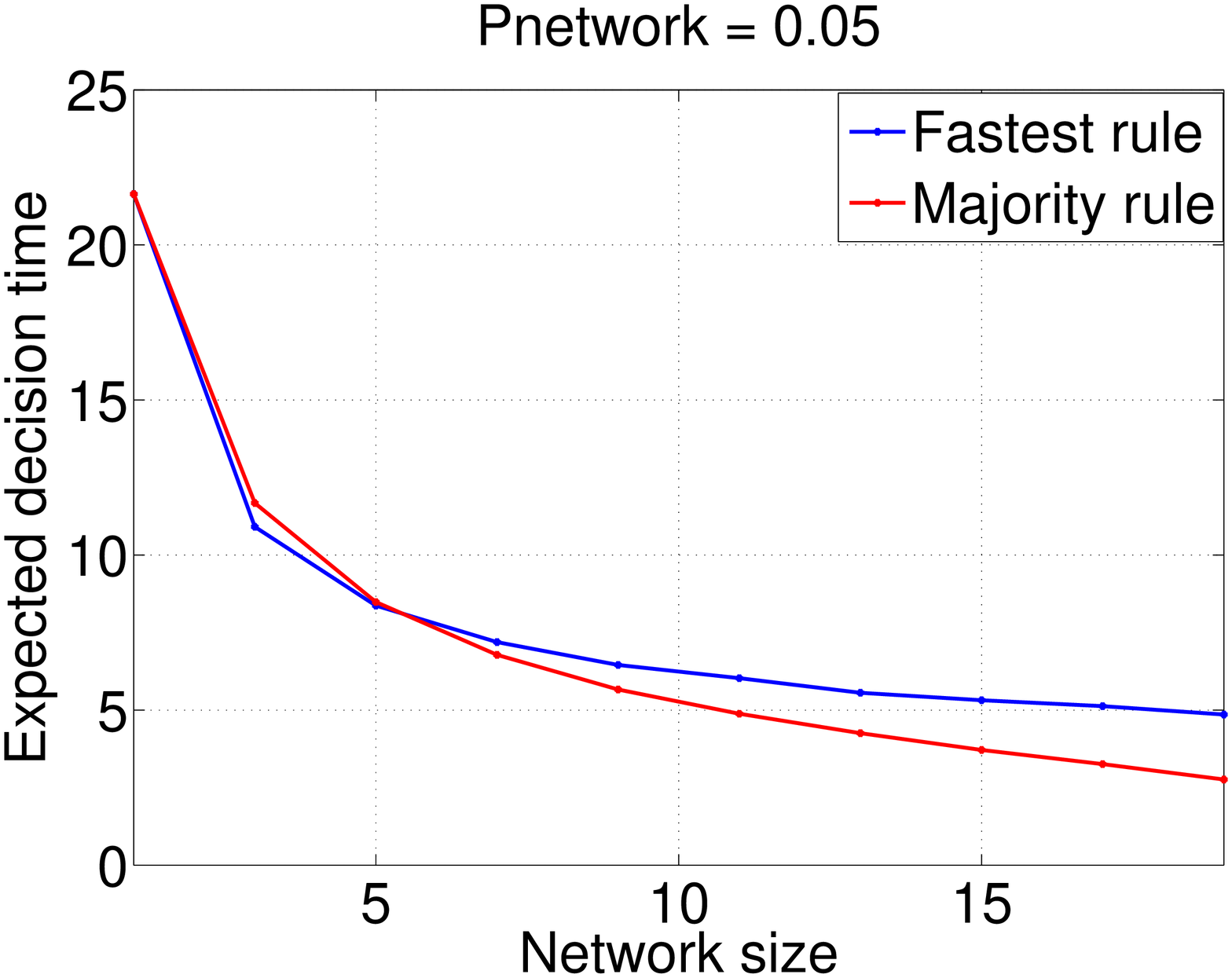}}
      {\includegraphics[width=.3\textwidth]{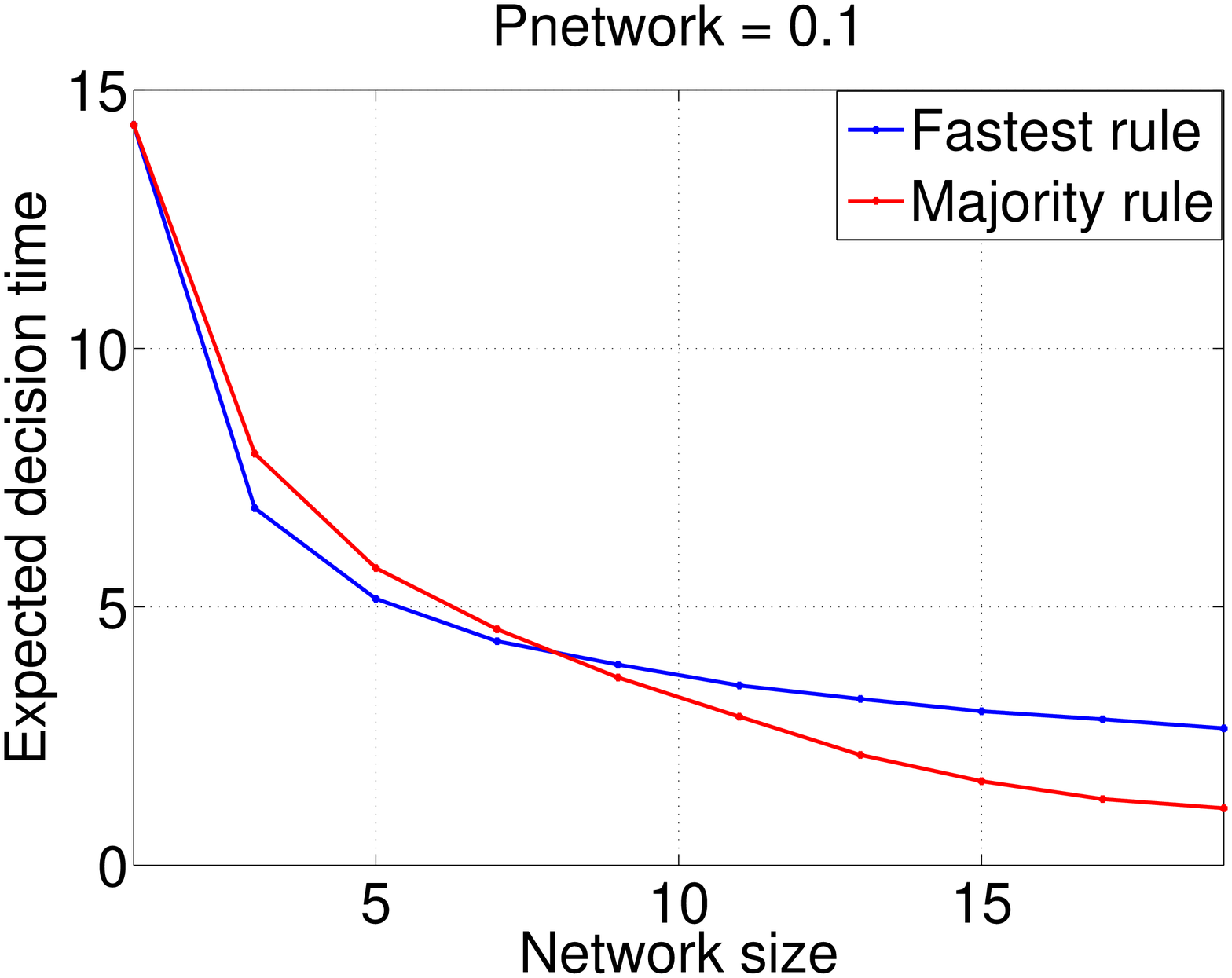}}
      \caption{Expected decision time for the \emph{fastest} and the
        \emph{majority} rules versus group size $N$, for various network
        accuracy levels.}
      \label{fig-netaccDT}
    \end{center}
  \end{figure}
  
  We summarize our observations about which rule is optimal (i.e., which
  rule requires the least number of observations) as follows:
  \begin{enumerate}
  \item the optimal rule varies with the desired network accuracy, at fixed
    network size;
  \item the optimal rule varies with the desired network size, at fixed
    network accuracy; and
  \item the change in optimality occurs at different network sizes for
    different accuracies.
  \end{enumerate}
\end{example}

\subsection{Decision making in cognitive psychology  revisited}\label{sec:DMCPR}

In this section we highlight some interesting relationships between our
results in sequential decision aggregation (SDA) and some recent
observations about mental behavior from the cognitive psychology
literature. Starting with the literature review in
Subsection~\ref{ssection-cog}, our discussion here is based upon the
following assumptions:
\begin{enumerate}
\item SDA models multi-modal integration in cognitive information
  processing (CIP),
\item the number of SDMs correspond to the number of sensory modalities in
  CIP,
\item the expected decision time in the SDA setup is analogous to the
  reaction time in CIP, and
\item the decision probability in the SDA setup is analogous to the firing
  rate of neurons in CIP.
\end{enumerate}%
Under these assumptions, we relate our SDA analysis to four recent
observations reported in the CIP literature. In short, the \emph{fastest}
and \emph{majority} rules appear to emulate behaviors that are similar to
the ones manifested by the brain under various conditions. These
correspondences are summarized in Table~\ref{table-cog-eng} and described
in the following paragraphs.
\begin{table}[!h]
\begin{center}
  \begin{tabular}{||c||c|c||}
    \hline
    $$
     \emph{Multi-sensory integration sites (cognitive psychology)}  &
    \emph{Sequential decision aggregation (engineering sciences)}  \\
    \hline\hline
    Suppressive behavior of firing rate
        & Decision probability decreases with increasing $N$
    \\[-1.5em]&
    \\ \hline
     Sub-additive behavior of firing rates
        & Probability of decision slightly increases with increasing $N$  \\
    [-1.5em]&
    \\ \hline
     Additive behavior of firing rates
    & Decision probability linearly increases with increasing $N$ \\
   [-1.5em] &
    \\ \hline
         Super-additive behavior of firing rates
    & Decision probability exponentially increases with increasing $N$\\
   [-1.5em] &
    \\ \hline
  \end{tabular} 
\end{center}
\caption{Cognitive Psychology - Engineering dictionary of correspondences.
  \label{table-cog-eng}} 
\end{table}




First, we look at the observation in CIP that multi-modal sites can exhibit
suppressive behaviors (first row in Table~\ref{table-cog-eng}). We find
that suppressive behavior is not contradictory with the nature of such a
site. Indeed, Proposition~\ref{prop:A-ET-Fastest} describes situations
where an increased group size degrades the decision accuracy of a group
using the \emph{fastest} rule.

Second, we look at the observation in CIP that, for some high-intensity
stimuli, the firing rate of multi-modal integration sites is similar to the
firing rate of uni-modal integration sites (second row in
Table~\ref{table-cog-eng}). This similarity behavior appears related to
behaviors observed in Figure~\ref{fig-netaccDT}.  The second and third
plots in Figure~\ref{fig-netaccDT} illustrate how, in small groups with
high individual accuracy and relatively low group decision accuracy, the
\emph{fastest} rule is optimal.  Since a multi-model integration site
implementing a fastest aggregation rule behaves similarly to a uni-modal
integration site, our result give a possible optimality interpretation of
the observed ``multi-modal similar to uni-modal'' behavior.

Third, we look at the observation in CIP that activation of multi-modal
integration sites is often accompanied with an increase in the accuracy as
compared to the accuracy of a uni-sensory integration site (third and forth
rows in Table~\ref{table-cog-eng}). The first plot in
Figure~\ref{fig-netaccDT} shows that when the required performance is a
high accuracy, the majority rule is better than the fastest. Indeed
Proposition~\ref{prop:Acc_Maj} proves that, for the \emph{majority} rule,
the accuracy monotonically increases with the group size, sometimes
exponentially.

Fourth, we look at the observation in CIP that, even under the same type of
stimuli, the stimuli strength affects the additivity of the neuron firing,
i.e., the suppressive, additive, sub-additive or super-additive behavior of
the firing rates.  Additionally, scientists have observed that depending on
the intensity of the stimuli, various areas of the brain are activated when
processing the same type of stimuli~\cite{SW-UN:10,PL-TP-TS-MW-BS:05,
  SW-UN:09, RB-EB-etal:06}.  A possible explanation for these two observed
behaviors is that the brain processes information in a way that maintains
optimality.  Indeed, our comparison in the middle and right parts of
Figure~\ref{fig-netaccDT} shows how the fastest rule is optimal when
individual SDMs are highly accurate (strong and intact stimuli) and, vice
versa, the majority rule is optimal when individual SDMs are relatively
inaccurate (weak and degraded stimuli).

We observed in the middle and right part of Figure~\ref{fig-netaccDT} that,
for high individual accuracies, the \emph{fastest} rule is more efficient
than the \emph{majority} rule. We reach this conclusion by noting two
observations: first, smaller group sizes require higher local accuracies
than larger group sizes in order to maintain the same group accuracy;
second, the \emph{fastest} rule is optimal for small groups while the
\emph{majority} rule is always optimal for larger groups.  

\section{Conclusion}\label{Sec-conc}
In this work, we presented a complete analysis of how a group of SDMs can
collectively reach a decision about the correctness of a hypothesis. We
presented a numerical method that made it possible to completely analyze
and understand interesting fusion rules of the individuals decisions. The
analysis we presented concentrated on two aggregation rules, but a similar
analysis can be made to understand other rules of interest. An important
question we were able to answer, was the one relating the size of the group
and the overall desired accuracy to the optimal decision rules. We were
able to show that, no single rule is optimal for all group sizes or for
various desired group accuracy. We are currently extending this work to
cases where the individual decision makers are not identical.

\bibliographystyle{ieeetr}
\bibliography{alias,Main,FB}

\appendix

\subsection{Asymptotic and monotonicity results on combinatorial sums}
\label{appendix:combinatorial-sums}

Some of the results provided for the \emph{fastest} rule and for the
\emph{majority} rule are based on the following properties of the binomial
expansion $(x+y)^N=\sum_{j=0}^{N} {N \choose j} x^jy^{N-j}$.

\begin{lemma}[Properties of half binomial expansions]
  \label{Lemm-lim-comb-sum}
  For an odd number $N\in\natural$, and for real numbers $c\in\real$ and
  $x\in\real$ satisfying $0 < c \leq 1$ and $0\leq x \leq c/2$, define
  $$
  \underline{S}(N;c,x)=\sum_{j=0}^{\lfloor N/2 \rfloor}{N \choose j} x^j 
  (c-x)^{N-j} \qquad \text{and} \qquad \overline{S}(N;c,x)=\sum_{j=\lceil
    N/2 \rceil}^{N}{N \choose j} x^j (c-x)^{N-j}.
  $$
  The following statements hold true:
  \begin{enumerate}
  \item \label{fact:LemFact1} if $0\leq x < c/2$, then, taking limits over
    odd values of $N$,
    \begin{equation*}
      \lim_{N \to \infty} \frac{\underline{S}(N;c,x)}{c^N}=1 \qquad 
      \text{and} \qquad \lim_{N \to \infty}
      \frac{\overline{S}(N;c,x)}{c^N}=0;
    \end{equation*}

  \item\label{fact:LemFact2} if $x=c/2$, then
    \begin{equation*}
      \underline{S}(N;c,x) = \overline{S}(N;c,x) = \frac{c^N}{2};
    \end{equation*}

  \item\label{fact:LemFact3} if $c=1$ and $0\leq x <1/2$, then
    \begin{equation*}
      \overline{S}(N+2; 1, x) < \overline{S}(N; 1, x)  \qquad \text{and}
      \qquad \underline{S}(N+2; 1, x) > \underline{S}(N; 1, x). 
    \end{equation*}
  \end{enumerate}
\end{lemma}
\begin{proof}
  To prove statement~\ref{fact:LemFact1}, we start with the obvious
  equality $c^N=(c-x+x)^N=\underline{S}(N;c,x)+
  \overline{S}(N;c,x)$. Therefore, it suffices to show that $\lim_{N
    \to\infty} \frac{\overline{S}(N;c,x)}{c^N}=0$.  Define the shorthand
  $h(j):= {N \choose j} x^j (c-x)^{N-j}$ and observe
  \begin{equation*}
    \frac{h(j)}{h(j+1)} = \frac{\frac{N!}{j! (N-j)!}x^j
      (c-x)^{N-j}}{\frac{N!}{(j+1)!
        (N-j-1)!}x^{j+1}(c-x)^{N-j-1}}=\frac{j+1}{N-j} \frac{c-x}{x}. 
  \end{equation*}
  It is straightforward to see that $\frac{h(j)}{h(j+1)} > 1 \iff c j -x N
  + c - x > 0 \iff j > \frac{xN}{c}-\frac{(c-x)}{c}$. Moreover, if $j
  >\frac{N}{2}$ and $0 \leq x < \frac{c}{2}$, then $j-\frac{x
    N}{c}+\frac{c-x}{c} > \frac{N}{2} -\frac{xN}{c} + \frac{c-x}{c} \geq
  \frac{N}{2} -\frac{N}{2} + \frac{c-x}{c}> 0$. Here, the second inequality
  follows from the fact that $-\frac{xN}{c}\geq -\frac{N}{2}$ if $0\leq x
  <\frac{c}{2}$. In other words, if $j >\frac{N}{2}$ and $0 \leq x <
  \frac{c}{2}$, then $\frac{h(j)}{h(j+1)} > 1$.  This result implies the
  following chain of inequalities $f\left(\lceil N/2
    \rceil\right)>f\left(\lceil N/2 \rceil +1\right)>\dots >h(N)$ providing
  the following bound on $\overline{S}(N;c,x)$
  \begin{align*}
    \overline{S}(N;c,x)&=\frac{\sum_{j=\lceil N/2 \rceil}^{N}{N \choose j}
      x^j (c-x)^{N-j}}{c^N}< \frac{\lceil N/2 \rceil {N \choose \lceil {N/2}
        \rceil} x^{\lceil N/2 \rceil} (c-x)^{\lfloor N/2 \rfloor}}{c^N}.
  \end{align*}
  Since ${N \choose \lceil {N/2} \rceil} < 2^N$, we can write
  \begin{align*}
    \overline{S}(N;c,x)&< \ceil{{N}/{2}} \frac{2^N x^{\lceil N/2 \rceil}
      (c-x)^{\lfloor N/2 \rfloor}}{c^N} = \lceil N/2 \rceil
    \left(\frac{2x}{c}\right)^{\lceil N/2 \rceil}
    \left(\frac{2(c-x)}{c}\right)^{\lfloor N/2 \rfloor}  \\
    & = \lceil N/2 \rceil \left(\frac{2x}{c}\right) \left(\frac{2x}{c}
    \right)^{\lfloor N/2 \rfloor} \left( \frac{2(c-x)}{c}\right)^{\lfloor
      N/2 \rfloor}.
  \end{align*}

  Let $\alpha = \frac{2x}{c}$ and $\beta=2\left(\frac{c-x}{c} \right)$ and
  consider $ \alpha \cdot \beta = \frac{4x (c-x)}{c^2}$.  One can easily
  show that $\alpha \cdot \beta < 1$ since $4cx - 4x^2 - c^2 = - (c-2x)^2 <
  0$.  The proof of statement~\ref{fact:LemFact1} is completed by noting
  \begin{equation*}
    \lim_{N \to \infty} \overline{S}(N;c,x) <  \lim_{N \to \infty} \lceil N/2
    \rceil \left( \frac{2x}{c} \right) \left( \alpha \cdot \beta
    \right)^{\lfloor N/2 \rfloor} = 0. 
  \end{equation*}

  The proof of the statement~\ref{fact:LemFact2} is straightforward. In
  fact it follows from the symmetry of the expressions when $x =
  \frac{c}{2}$, and from the obvious equality $\sum_{j=0}^{N} {N \choose j}
  x^j (c-x)^{N-j} = c^N$.

  Regarding statement~\ref{fact:LemFact3}, we prove here only that
  $\overline{S}(N+2; 1, x) < \overline{S}(N; 1, x)$ for $0\leq x <1/2$. The
  proof of $ \underline{S}(N+2; 1, x) > \underline{S}(N; 1, x)$ is
  analogous.  Adopting the shorthand
  \begin{equation*}
  f(N,x):=\sum_{i=\lceil\frac{N}{2}\rceil}^N {N \choose i} x^i (1-x)^{N-i},    
  \end{equation*}
  we claim that the assumption $0<x<1/2$ implies
  $$
  \Delta(N,x):=f(N+2,x)-f(N,x) < 0.
  $$
  To establish this claim, it is useful to analyze the derivative of
  $\Delta$ with respect to $x$.  We compute
  \begin{equation}
    \label{eq:Der}
    \frac{\partial f}{\partial x}(N,x)=\sum_{i=\lceil N/2 \rceil}^{N-1}i
    {N\choose i} x^{i-1}(1-x)^{N-i}-\sum_{i=\lceil N/2
      \rceil}^{N-1}(N-i){N\choose i}x^i(1-x)^{N-i-1}+Nx^{N-1}.
  \end{equation}
  The first sum $\sum_{i=\lceil N/2 \rceil}^{N-1}i {N\choose i}
  x^{i-1}(1-x)^{N-i}$ in the right-hand side of~\eqref{eq:Der} is equal to
  \begin{equation*}
    {N\choose \lceil N/2 \rceil} \Bigceil{\frac{N}{2}}
    x^{\lceil N/2 \rceil-1}\left(1-x\right)^{N-\lceil N/2
      \rceil}+\sum_{i=\lceil N/2 \rceil+1}^{N-1}i {N\choose i}
    x^{i-1}(1-x)^{N-i}.
  \end{equation*}
  Moreover, exploiting the identity $(i+1) {N\choose i+1}=(N-i) {N\choose
    i}$,
  \begin{align*}
    \sum_{i=\lceil N/2 \rceil+1}^{N-1}i {N\choose i} x^{i-1}(1-x)^{N-i}&=
    \sum_{i=\lceil N/2 \rceil}^{N-2}(i+1) {N\choose i+1}
    x^{i}(1-x)^{N-i-1} \\ 
    &=\sum_{i=\lceil N/2 \rceil}^{N-2} (N-i) {N\choose i} x^{i}(1-x)^{N-i-1}.
  \end{align*}
  The second sum in the right-hand side of~\eqref{eq:Der} can be rewritten
  as
  $$
  \sum_{i=\lceil N/2 \rceil}^{N-1}(N-i){N \choose
    i}x^i(1-x)^{N-i-1}=\sum_{i=\lceil N/2 \rceil}^{N-2}(N-i){N\choose
    i}x^i(1-x)^{N-i-1}+Nx^{N-1}.
  $$
  Now, many terms of the two sums cancel each other out and one can easily
  see that
  \begin{align*}
    \frac{\partial f}{\partial x}(N,x)&={N\choose \lceil N/2 \rceil} \lceil
    N/2 \rceil x ^{\lceil N/2 \rceil-1}\left(1-x\right)^{N-\lceil N/2
      \rceil}={N\choose \lceil N/2 \rceil} \lceil N/2 \rceil
    \left(x\left(1-x\right)\right)^{ \lceil N/2 \rceil-1},
  \end{align*}
  where the last equality relies upon the identity $N-\lceil N/2 \rceil =
  \lfloor N/2 \rfloor=\lceil N/2 \rceil-1$.  Similarly, we have
  \begin{align*}
    \frac{\partial f}{\partial x}(N+2,x)&={N+2\choose \lceil N/2 \rceil+1}
    \left(\lceil N/2 \rceil+1\right) \left(x\left(1-x\right)\right)^{\lceil
      N/2 \rceil}.
  \end{align*}
  Hence  
  \begin{align*}
    \frac{\partial \Delta}{\partial
      x}(N,x)&=\left(x\left(1-x\right)\right)^{ \lceil N/2 \rceil-1}
    \left({N+2\choose \lceil N/2 \rceil+1} \left(\lceil N/2 \rceil+1\right)
      x(1-x)- {N\choose \lceil N/2 \rceil} \lceil N/2 \rceil \right).
  \end{align*}
  Straightforward manipulations show that
  $$
  {N+2\choose \lceil N/2 \rceil+1} \left(\lceil N/2 \rceil+1\right)=4
  \frac{N+2}{N+1} \lceil N/2 \rceil {N\choose \lceil N/2 \rceil},
  $$
  and, in turn, 
  \begin{multline*}
    \frac{\partial \Delta}{\partial x}(N,x) = {N\choose \lceil N/2 \rceil}
    \ceil{\frac{N}{2}} \left(x\left(1-x\right)\right)^{ \lceil N/2
      \rceil-1} \left[ 4\frac{N+2}{N+1} x(1-x)- 1 \right] \\
    =: g(N,x)\left[ 4\frac{N+2}{N+1} x(1-x)- 1 \right],
  \end{multline*}
  where the last equality defines the function $g(N,x)$. Observe that $x>0$
  implies $g(N,x)>0$ and, otherwise, $x=0$ implies $g(N,x)=0$.  Moreover,
  for all $N$, we have that $f(N,1/2)=1/2$ and $f(N,0)=0$ and in turn that
  $\Delta(N,1/2)=\Delta(N,0)=0$. Additionally
  $$
  \frac{\partial \Delta}{\partial x}(N,1/2) =
  g(N,1/2)\left(\frac{N+2}{N+1}-1\right)>0
  $$
  and
  $$
  \frac{\partial \Delta}{\partial x}(N, 0)=0 \qquad \text{and} \qquad
  \frac{\partial \Delta}{\partial x}(N, 0^+)= g(N,0^+)\left(0^+-1\right)<0.
  $$
  The roots of the polynomial $x\mapsto 4\frac{N+2}{N+1} x(1-x)- 1$ are
  $\frac{1}{2}\left(1\pm\sqrt{\frac{1}{N+2}}\right)$, which means that the
  polynomial has one root inside the interval $(0,1/2)$ and one inside the
  interval $(1/2, 1)$. Considering all these facts together, we conclude
  that the function $x\mapsto\Delta(N,x)$ is strictly negative in $(0,
  1/2)$ and hence that $f(N+2,x)-f(N,x) < 0$.
\end{proof}

\subsection{Computation of the decision probabilities for a single SDM
  applying the SPRT test}
\label{subsec:AccTime}

In this appendix we discuss how to compute the probabilities
\begin{equation}\label{eq:prob}
  \left\{\pndz\right\} \cup \left\{ \pzz (t), \poz(t)\right\}_{t\in
    \N}\qquad \text{and} \left\{\pndo\right\}\cup \left\{\pzo(t),
    \poo(t)\right\}_{t\in\N}
\end{equation} 
for a single SDM applying the classical \emph{sequential probability ratio
  test} (SPRT). For a short description of the SPRT test and for the
relevant notation, we refer the reader to Section~\ref{SSec-sim-nocomm}.
We consider here observations drawn from both discrete and continuous
distributions.

\subsubsection{Discrete distributions of the Koopman-Darmois-Pitman form}
\label{subsec:AccTimeDiscrete}
This subsection review the procedure proposed in \cite{LY:94} for a certain
class of discrete distributions. Specifically, \cite{LY:94} provides a
recursive method to compute the exact values of the
probabilities~\eqref{eq:prob}; the method can be applied to a broad class
of discrete distributions, precisely whenever the observations are modeled
as a discrete random variable of the Koopman-Darmois-Pitman form.

With the same notation as in Section~\ref{SSec-sim-nocomm}, let $X$ be a
discrete random variable of the Koopman-Darmois-Pitman form; that is
$$
f(x,\theta)=
\begin{cases}
  h(x)\exp(B(\theta) Z(x) -A(\theta)),\quad & \text{if }  x\in \mathcal{Z},\\
  0, &\text{if } x\notin \mathcal{Z},
\end{cases}
$$
where $h(x)$, $Z(x)$ and $A(\theta)$ are known functions and where
$\mathcal{Z}$ is a subset of the integer numbers $\Z$. In this section we
shall assume that $Z(x)=x$. Bernoulli, binomial, geometric, negative
binomial and Poisson distributions are some widely used distributions of
the Koopman-Darmois-Pitman form satisfying the condition $Z(x)=x$. For
distributions of this form, the likelihood associated with the $t$-th
observation $x(t)$ is given by
$$
\lambda(t)=(B(\theta_1)-B(\theta_0))x(t)-(A(\theta_1)-A(\theta_0)).
$$
Let $\eta_0, \eta_1$ be the pre-assigned thresholds. Then, one can see that
sampling will continue as long as
\begin{equation}
  \label{eq-KDP}
  \frac{\eta_0+t(A(\theta_1)-A(\theta_0))}{B(\theta_1)-B(\theta_0))} 
  <\sum_{i=1}^t x(i)<
  \frac{\eta_1+t(A(\theta_1)-A(\theta_0))}{B(\theta_1)-B(\theta_0))} 
\end{equation}
for $B(\theta_1)-B(\theta_0)>0$; if $B(\theta_1)-B(\theta_0)<0$ the
inequalities would be reversed. Observe that $\sum_{i=1}^t x(i)$ is an
integer number. Now let $\bar{\eta}_{0}^{(t)}$ be the smallest integer
greater than
$\left\{\eta_0+t(A(\theta_1)-A(\theta_0))\right\}/(B(\theta_1)-B(\theta_0))$
and let $\bar{\eta}_{1}^{(t)}$ be the largest integer smaller than
$\left\{\eta_1+t(A(\theta_1)-A(\theta_0))\right\}/(B(\theta_1)-B(\theta_0))$. Sampling
will continue as long as $ \bar{\eta}_{0}^{(t)} \leq \mathcal{X}(t)\leq
\bar{\eta}_{1}^{(t)} $ where $\mathcal{X}(t)=\sum_{i=1}^tx(i)$. Now suppose
that, for any $\ell \in [\bar{\eta}_{0}^{(t)}, \bar{\eta}_{1}^{(t)}]$ the
probability $\P[\mathcal{X}(t)=\ell]$ is known. Then we have
\begin{equation*}
  \P[\mathcal{X}(t+1)=\ell|H_i] = 
  \sum_{j=\bar{\eta}_{0}^{(t)}}^{\bar{\eta}_{1}^{(t)}}f(\ell-j;\theta_i)\P[\mathcal{X}(t)=j|H_i], 
\end{equation*}
and
\begin{align*}
  \poi(t+1)&=\sum_{j=\bar{\eta}_{0}^{(t)}}^{\bar{\eta}_{1}^{(t)}}\,\,\sum_{r=\bar{\eta}_{1}^{(t)}-j+1}^\infty
  \P[\mathcal{X}(t)=j | H_i]  f(r;\theta_i),\nonumber\\ 
  p_{0|i}(t+1)&=\sum_{j=\bar{\eta}_{0}^{(t)}}^{\bar{\eta}_{1}^{(t)}}\,\,\sum_{r=-\infty}^{\bar{\eta}_{0}^{(t)}-j-1}
  \P[\mathcal{X}(t)=j | H_i] f(r;\theta_i).
\end{align*}
Starting with $\P[\mathcal{X}(0)=1]$, it is possible to compute recursively
all the quantities $\left\{\pij(t)\right\}_{t=1}^\infty$ and
$\P[\mathcal{X}(t)=\ell]$, for any $t\in \N$, $\ell \in
[\bar{\eta}_{0}^{(t)}, \bar{\eta}_{1}^{(t)}]$, and
$\left\{\pij(t)\right\}_{t=1}^\infty$. Moreover, if the set
$\mathcal{Z}$ is finite, then the number of required computations is
finite.

\subsubsection{Computation of  accuracy and decision time for pre-assigned
  thresholds $\eta_0$ and $\eta_1$: continuous
  distributions}\label{subsec:AccTimeCont} 

In this section we assume that $X$ is a continuous random variable with
density function given by $f(x,\theta)$. As in the previous subsection,
given two pre-assigned thresholds $\eta_0$ and $\eta_1$, the goal is to
compute the probabilities $\pij(t)=\P[\text{say} H_i|H_j, T=t]$, for
$i,j \in\{1,2\}$ and $t\in\N$.

We start with two definitions. Let $f_{\lambda, \theta_i}$ and
$f_{\Lambda(t), \theta_i}$ denote, respectively, the density function of
the log-likelihood function $\lambda$ and of the random variable
$\Lambda(t)$, under the assumption that $H_i$ is the correct hypothesis.
Assume that, for a given $t\in\N$, the density function $f_{\Lambda(t),
  \theta_i}$ is known. Then we have 
\begin{equation*} 
  f_{\Lambda(t), \theta_i}(s)=\int_{\eta_0}^{\eta_1} f_{\lambda,
    \theta_i}(s-x)f_{\Lambda(t), \theta_i}(x)dx,\qquad s\in \left(\eta_0,
    \eta_1\right), 
\end{equation*}
and
\begin{align*}
  \poi(t) &=\int_{\eta_0}^{\eta_1}\!\!\left(\int_{\eta_1-x}^{\infty}
    f_{\lambda, \theta_i}(z)dz \right)f_{\Lambda(t), \theta_i}(x)dx,
  \text{ and }
  p_{0|i}(t)=\int_{\eta_0}^{\eta_1}\!\!\left(\int_{-\infty}^{\eta_0-x}
    f_{\lambda, \theta_i}(z)dz \right)f_{\Lambda(t), \theta_i}(x)dx.
\end{align*}
In what follows we propose a method to compute these quantities
based on a uniform discretization of the functions $\lambda$ and
$\Lambda$. Interestingly, we will see how the classic SPRT algorithm can be
conveniently approximated by a suitable absorbing Markov chain and how,
through this approximation, the probabilities
$\left\{\pij(t)\right\}_{t=1}^\infty$, $i,j\in\{1,2\}$, can be
efficiently computed. Next we describe our discretization approach.

First, let $\delta\in \realpositive$, $\bar{\eta}_0=\lfloor
\frac{\eta_0}{\delta} \rfloor\delta$ and
$\bar{\eta}_1=\lceil\frac{\eta_1}{\delta} \rceil\delta$. Second, for
$n=\lceil\frac{\eta_1}{\delta}\rceil- \lfloor \frac{\eta_0}{\delta}\rfloor
+1$, introduce the sets
\begin{equation*}
  \mathcal{S}=\left\{s_1,\ldots, s_n\right\}\quad \text{and} \quad
  \Gamma=\left\{\gamma_{-n+2},
    \gamma_{-n+3},\ldots,\gamma_{-1},\gamma_0,\gamma_1,\dots,
    \gamma_{n-3},\gamma_{n-2}\right\},  
\end{equation*}
where $s_i=\bar{\eta}_0+(i-1)\delta$, for $i\in \until{n}$, and
$\gamma_{i}=i\delta$, for $i\in\left\{-n+2, -n+3,\ldots, n-3,n-2\right\}$.
Third, let $\bar{\lambda}$ (resp. $\bar{\Lambda}$) denote a discrete random
variable (resp. a stochastic process) taking values in $\Gamma$ (resp. in
$\mathcal{S}$). Basically $\bar{\lambda}$ and $\bar{\Lambda}$ represent the
discretization of $\Lambda$ and $\lambda$, respectively. To characterize
$\bar{\lambda}$, we assume that
$$
\P\left[\bar{\lambda}=i\delta\right]= \P\left[i\delta-\frac{\delta}{2}\leq \lambda \leq i\delta+\frac{\delta}{2}\right], \qquad i \in \left\{-n+3,\ldots,n-3\right\},
$$
and
$$
\P\left[\bar{\lambda}=(-n+2)\delta\right]= \P\left[ \lambda \leq (-n+2)\delta+\frac{\delta}{2}\right] \qquad \text{and} \qquad \P\left[\bar{\lambda}=(n-2)\delta\right]= \P\left[ \lambda \geq (n-2)\delta-\frac{\delta}{2}\right].
$$
From now on, for the sake of simplicity, we shall denote
$\P\left[\bar{\lambda}=i\delta\right]$ by $p_i$.  Moreover we adopt the
convention that, given $s_i\in \mathcal{S}$ and $\gamma_j\in \Gamma$, we
have that $s_i+\gamma_j:=\bar{\eta}_0$ whenever either $i=1$ or $i+j-1\leq
1$, and $s_i+\gamma_j:=\bar{\eta}_1$ whenever either $i=n$ or $i+j-1\geq
n$. In this way $s_i+\gamma_j$ is always an element of $\mathcal{S}$. Next
we set $\bar{\Lambda}(t):=\sum_{h=1}^{t}\bar{\lambda}(h)$.

To describe the evolution of the stochastic process $\bar{\Lambda}$, define
the row vector $\pi(t)=[\pi_1(t),\ldots, \pi_n(t)]^T \in \R^{1\times n}$
whose $i$-th component $\pi_i(t)$ is the probability that $\bar{\Lambda}$
equals $s_i$ at time $t$, that is,
$\pi_i(t)=\P\left[\bar{\Lambda}(t)=s_i\right]$. The evolution of $\pi(t)$
is described by the absorbing Markov chain $(\mathcal{S}, A, \pi(0))$ where
\begin{itemize}
\item $\mathcal{S}$ is the set of states with $s_1$ and $s_n$ as absorbing
  states;
\item $A=[a_{ij}]$ is the transition matrix: $a_{ij}$ denote the
  probability to move from state $s_i$ to state $s_j$ and satisfy,
  according to our previous definitions and conventions,
  \begin{itemize}
  \item $a_{11}=a_{nn}=1;\qquad a_{1i}=a_{nj}=0,\,\,\,\text{for}\,\,\,
    i\in\left\{2,\ldots,n\right\}\,\,\,\text{and}\,\,\, j\in\until{n-1}$;
  \item $a_{i1}=\sum_{s=-n+2}^{-h+1}p_s\,\,\, \text{and}\,\,\,
    a_{in}=\sum_{s=1}^{n-2}p_s,\,\,\,  h\in\left\{2,\ldots,n-1\right\}$; 
  \item $a_{ij}=p_{j-i}\,\,\,\,\,\, i,j\in\left\{2,\ldots,n-1\right\}$;
  \end{itemize}
\item $\pi(0)$ is the initial condition and has the property that
  $\P[\bar{\Lambda}(0)=0]=1$.
\end{itemize}
In compact form we write $\pi(t)=\pi(0)A^t$.

The benefits of approximating the classic SPRT algorithm with an absorbing
Markov chain $(\mathcal{S}, A, \pi(0))$ are summarized in the next
Proposition. Before stating it, we provide some useful definitions.
First, let $Q\in \R^{(n-2)\times (n-2)}$ be the matrix obtained by deleting
the first and the last rows and columns of $A$. Observe that $I-Q$ is an
invertible matrix and that its inverse $F:=(I-Q)^{-1}$ is typically known
in the literature as the \emph{fundamental matrix} of the absorbing matrix
$A$.
Second let $A^{(1)}_{2:n-1}$ and $A^{(n)}_{2:n-1}$ denote, respectively,
the first and the last column of the matrix $A$ without the first and the
last component, i.e., $A^{(1)}_{2:n-1}:=\left[a_{2,1}, \ldots,
  a_{n-1,1}\right]^T$ and $A^{(n)}_{2:n-1}:=\left[a_{2,n},\ldots, a_{n-1,
    n}\right]^T$.
Finally, let $e_{\lfloor \frac{\eta_0}{\delta} \rfloor+1}$ and $\1_{n-2}$
denote, respectively, the vector of the canonical basis of $\R^{n-2}$
having $1$ in the $(\lfloor \frac{\eta_0}{\delta} \rfloor+1)$-th position
and the $(n-2)$-dimensional vector having all the components equal to $1$
respectively.

\begin{proposition}[SPRT as a Markov Chain]
  Consider the classic SPRT test. Assume that we model it through the
  absorbing Markov chain $(\mathcal{S}, A, \pi(0))$ described above. Then
  the following statements hold:
  \begin{enumerate}
  \item $p_{0|j}(t)=\pi_1(t)-\pi_1(t-1)$ and
    $p_{1|j}(t)=\pi_n(t)-\pi_n(t-1)$, for $t\in\N$;
  \item $\P[\text{say } H_0|H_j]=e^T_{\lfloor \frac{\eta_0}{\delta}
      \rfloor+1}N\bar{a}_1$ and $\P[\text{say } H_0|H_j]=e^T_{\lfloor
      \frac{\eta_0}{\delta} \rfloor+1}N\bar{a}_n$; and
  \item $\E[T|H_j]=e^T_{\lfloor \frac{\eta_0}{\delta} \rfloor+1}F\1_{n-2}$.
  \end{enumerate}
\end{proposition}

\end{document}